\documentclass[ba, preprint]{imsart}

\RequirePackage{amsthm,amsmath,amsfonts,amssymb}
\RequirePackage[round]{natbib}
\RequirePackage[colorlinks,citecolor=blue,urlcolor=blue,backref=page,backref=page]{hyperref}
\RequirePackage{graphicx}
\usepackage{subcaption}
\usepackage{booktabs}
\usepackage{multirow}

\pubyear{2026}
\arxiv{2010.00000}
\volume{TBA}
\issue{TBA}
\firstpage{1}
\lastpage{1}

\startlocaldefs
\theoremstyle{plain}

\theoremstyle{definition}

\theoremstyle{remark}

\theoremstyle{plain}
\newtheorem{thm}{Theorem}

\theoremstyle{definition}
\newtheorem{assump}{Assumption}
\newtheorem{remark}{Remark}
\newcommand{\E}{\mathbb{E}}
\newcommand{\Pp}{\mathbb{P}}
\newcommand{\R}{\mathbb{R}}
\newcommand{\op}{o_{\Pp}}
\newcommand{\Op}{O_{\Pp}}

\endlocaldefs

\begin{document}

\begin{frontmatter}
\title{Inverse Probability Weighting in a Post-Bayesian World}
\runtitle{IPW in a Post-Bayesian World}

\begin{aug}

\author[A]{\fnms{Owen}~\snm{Thomas}\ead[label=e1]{o.m.t.thomas@medisin.uio.no}\orcid{0000-0000-0000-0000}},
\author[A,C]{\fnms{William}~\snm{Denault}\ead[label=e2]{william.denault@medisin.uio.no}\orcid{0000-0000-0000-0000}}
\author[C]{\fnms{Johan}~\snm{de Aguas}\ead[label=e3]{johanmd@math.uio.no}\orcid{0000-0000-0000-0000}}
\and
\author[B,C]{\fnms{Valeria}~\snm{Vitelli}\ead[label=e4]{valeria.vitelli@medisin.uio.no}\orcid{0000-0000-0000-0000}}
\address[A]{Oslo Centre for Biostatistics and Epidemiology, Oslo University Hospital, Norway\printead[presep={,\ }]{e1,e2}}

\address[B]{Oslo Centre for Biostatistics and Epidemiology, Department of Biostatistics, University of Oslo, Norway\printead[presep={,\ }]{e2,e4}}

\address[C]{Integreat, the Norwegian Centre for Knowledge-driven Machine Learning, Norway\printead[presep={,\ }]{e2,e3,e4}}
\runauthor{O. Thomas et al.}
\end{aug}

\begin{abstract}
We present a rationale for the use of Inverse Probability Weighting (IPW) in a post-Bayesian framework, in which the bias-correction provided by IPW in a frequentist context is reframed as a retargetting of the Kullback-Leibler (KL) divergence between the statistical model and the true data-generating process. We provide a coherent argument in support of this approach, including theoretical results concerning convergence and properties of the generalised belief posteriors. We present examples demonstrating the utility of post-Bayesian IPW in practice: these include two simulated examples of inference under selection bias in the observed data, and a large-scale real-data example concerning systematic biases present in registry data when using prostate-specific antigen (PSA) to predict prostate cancer mortality. The empirical and theoretical results together show the utility of IPW to address classes of problems previously intractable within a principled Bayesian approach.
\end{abstract}

\begin{keyword}[class=MSC]
\kwdgroup[type=primary]{\kwd{00X00}
\kwd{00X00}}
\kwdgroup[type=secondary]{\kwd{00X00}}
\end{keyword}

\begin{keyword}
\kwd{Generalised Bayes}
\kwd{Inverse Probability Weighting}
\end{keyword}

\end{frontmatter}

\section{Introduction}

In this article, we motivate the use of Inverse Probability Weighting\footnote{Also elsewhere referred to as importance weighting, importance reweighting, Horvitz-Thompson weighting, and other wordings.} (IPW) for Generalised Bayesian inference on population-level parameters as a tool for redefining the target inference distribution from the data generating process (d.g.p.) $G$ to a valid user-specified distribution $Q$. IPW has broad applications across many statistical challenges, including causal inference, missing data, attrition, and selection bias, among others \citep{IPWhorvitzThompson1952,IPWrobinsHernan2000,IPWcortes2008,IPWhuber2014,IPWdeAguas2025}. We pursue a generic approach here for introducing IPW into the post-Bayesian framework, while maintaining a running application example connected to selection bias. The motivation to redefine the data generating process is clear when dealing with instances of selection bias, where the observed data are drawn from a distribution that is different from the target population, potentially introducing a bias into the resulting posterior estimates. This situation is common in contemporary heterogeneous data environments: data quality comes at a cost, such that large-scale, cheaply acquired, potentially ``Big N'' data sets are available, in contrast to smaller-scale, bespoke, reduced bias ``Small N'' sets. One concrete example is in the world of biomedical research, which uses both population-scale registry data and clinical-scale patient data.

The paradigm in this article assumes that the researcher has access to a large-scale ``Big N'' data set $X$ drawn from a source distribution $G$ with density\footnote{including possible dependence on other independent variables $Z$, i.e. via $g(X|Z)$.} $g$, and a small-scale ``Small N'' data set $X'$ drawn from a target distribution $Q$ with density $q$. We propose using a classifier trained to distinguish these two data sets to estimate the density ratio $r_i = q(X_i)/g(X_i)$ evaluated at each data point $X_i$, which are then used as weighs $w_i$ to reweight the likelihood in a Bayesian update using $X$ under some assumed statistical likelihood $p(X|\theta)$. We show that this is equivalent to redefining the minimising Kullback-Leibler (KL) divergence of the Bayesian update from $KL(g(.)||p(.|\theta))$ to $KL(q(.)||p(.|\theta))$, i.e. projecting the inference procedure to target $Q$ instead of the biased $G$, under stated conditions. We demonstrate proofs for posterior convergence to the new value of $\theta^*$, and conditions on $G$ and $Q$ for successful ratio estimation. Empirical examples are provided for simulated and real data sets, and the further relevance of the work is discussed.

\section{Background}

IPW is not typically used in standard Bayesian inference as it has been difficult to justify modifying the prescribed likelihood-based Bayesian workflow. In particular, IPW uses information from the selection process to debias estimators. Yet under the standard Bayesian paradigm the parameters of such selection process are either irrelevant for the likelihood of the outcome or must be explicitly incorporated in a full joint model. The full Bayes option typically results in subobtimal frequentist properties \citep{robins2015discussion}.

Fully nonparametric Bayesian methods avoid this objection because they produce posteriors over empirical distributions, not over parametric components of a model. In this sense, they are more aligned to IPW. The Bayesian bootstrap \citep{rubin1981bayesian}, for instance, uses samples $\pi_i$ from a Dirichlet distribution to repeatedly reweight the likelihood: each reweighted likelihood is then optimised to generate samples from the posterior. Bayesian bootstrap type methods have been applied across different statistical challenges, including inference under misspecification \citep{lyddon2019general} and the Cut-Bayesian modularisation paradigm \citep{pompe2021asymptotics}. More generally, modular Bayesian workflows offer a way to incorporate IPW-style weighting while preventing likelihood feedback from distorting the selection model \citep{mccandless2011cutting}.

The weighted likelihood bootstrap \citep{newton1994approximate} generalises the Bayesian bootstrap to a generic collection of weights $w_i$ depending on the inferential target. \citep{saarela2015bayesian} made the connection with causal inference and IPW methods, combining the derived weights $w_i$ with the Bayesian bootstrap to enable Bayesian inference with separate modelling of causal effects. However, their approach relies on iterated optimisations and hence is limited in which inference methods can be used. Further related follow-up work from the same authors pursues posterior predictive causal inference \citep{saarela2015predictive} and Bayesian estimation of causal models that are ``doubly robust'' \citep{saarela2016bayesian}, i.e. either one of the generative model or IPW model can be misspecified without introducing bias.

Historically, IPW has been a way to handle selection bias in a frequentist context \citep{IPWhorvitzThompson1952}. In such contexts, estimators of quantities of interest are debiased through reweighting observations according to their sampling or selection probabilities. For example, selection bias can arise when estimating the mean $\mu$ of a random variable $X$ with finite expectation in a target population $P_1$ from observations collected in a different source population $P_0$. In such an instance, if the target distribution is absolutely continuous with respect to the source distribution, one can compute:

\begin{equation}
    \hat{\mu}_{IPWE} = \frac{1}{n} \sum_{i=1}^n X_i w_i
\end{equation}

where weights $w_i=p_1(X_i)/p_0(X_i)$ are computed as target-source density ratios to counteract the bias present in the standard estimate, often using logistic regression models. Previous work has also examined the use of likelihood weighting to address ``informative sampling''  in a Bayesian update, both in the case where the pairwise dependencies among units attenuate to 0 in the limit of the population size \citep{savitsky2016bayesian}, and in the more general unattenuated case \citep{williams2020bayesian}. This work on ``informative sampling'' is somewhat similar to the running example that we present in this article, but does not pursue the generalised or post-Bayesian approach that we advocate, and therefore we believe our proposal has the potential to be more widely impactful across other uses of IPW.

In recent years the literature on generalised Bayesian inference has developed rapidly, so that the various assumptions underpinning Bayes' rule have been relaxed and the implications for inference procedures have been explored \citep{knoblauch2022optimization,miller2019robust}. This has often been justified due to the ever-present question of model (mis-)specification in Bayes, but we would like to expand the applicability of generalised Bayes to new sources of inferential bias.

Bayes' rule has been generalised in various ways, from the initial use of power posteriors, in which the likelihood is raised to a scalar value shared for all data points \citep{grunwald2012safe}, to more flexible usage of generic losses and scoring rules beyond the negative log likelihood, following the groundbreaking work of \citep{bissiri2016general}. It is known that conventional Bayes rules implicitly minimize the KL-divergence between the true data generating process and the assumed statistical model \citep{walker2013bayesian}: however, the KL is known to place large amounts of weight on outliers in the distributions, leading to potential instability in the resulting posteriors. Consequently, divergences other than the KL have been explored for performing Bayesian updates, including the Total Variational Distance, the Hellinger divergence, and more general classes of $f$-divergences \citep{jewson2018principles}. Recent work has also specifically adapted generalized Bayesian reasoning to causal inference \citep{javurek2026generalized}, enabling the use of priors directly on parameters representing causal effects, which are related to the data through loss-based generalised likelihoods.

The generalised Bayes paradigm has also caught on in the simulation-based inference community, which often works with suboptimally-specified simulator models, and for which exact inference using traditional Bayesian updates has been impossible due to the intractability of the likelihood function. New inference procedures in such contexts have often relied on scoring-rule-based posterior \citep{pacchiardi2024generalized}, misspecification-robust discrepancies using simulations from the statistical model \citep{thomas2025misspecification}, focusing on the predictive properties of the belief distributions \citep{loaiza2021focused}, or the Maximum Mean Discrepancy bootstrap \citep{dellaporta2022robust}, among others \citep{matsubara2022robust}. The use of classifiers to perform estimation of density ratios of interest has also become common: these have included approximating the update ratio present in Bayes rule \citep{thomas2022likelihood}, to training full belief distributions using neural posterior estimation \citep{cranmer2020frontier}.

We see that there is a wealth of variation of Bayes' theorem presently being developed. It has never sounded less defensible to suggest that ``There is no Bayes but Bayes'' \citep{robins2015discussion}: the motto previously coined in response to \citep{saarela2015bayesian}, in a bid to represent the limited adaptability of Bayesian methods.

This article unifies the above research directions with the purpose of training classifiers to approximate the ratios appropriate for the reweighting of data points' likelihood contributions when performing a Bayesian update in the presence of inferential bias. It will be structured as follows. The Methods Section \ref{sec:methods} will include brief introductions to relevant concepts from the literature, followed by a description and justification of the methodology proposed. Theoretical results will be presented in Section \ref{sec:theory}. Results will be presented in Section \ref{sec:results}, including two simulated examples and a real data epidemiological analysis. Sections \ref{sec:discussion}, \ref{sec:conclusion}, and \ref{sec:acknowledgements} will discuss, conclude, and acknowledge, respectively. The appendices contain some pragmatic advice concerning using the ratio estimation in practice, and assumptions and proofs supporting the theoretical results.

\section{Methods}\label{sec:methods}

\subsection{Ratio Estimation}

We consider a setting in which the analyst seeks to perform a Bayesian update on a finite-dimensional population-level parameter using a large number $n$ of data points $X$ drawn from a distribution $G$ with density $g$. The distribution $G$ may differ from the target distribution $Q$ with density $q$, from which a smaller number $m$ of higher quality data points $X'$ are available. Classifiers have been used successfully to approximate ratios of distributions using their samples in many inference contexts: here, we use them to approximate the ratio of $r(X) = q(X)/g(X)$ using the observed data $X$ and $X'$ sampled from each distribution, and without either of $g$ or $q$ necessarily being directly tractable. This is performed in a cross-validated manner to enable the evaluation of the ratio on all of the data $X$ without reusing data between training and test: the subsequent evaluated ratios: $w_i = q(X_i)/g(X_i)$ consequently reflect the odds of each individual data point $X_i$ belonging to the target unbiased population relative to the source biased population, i.e. the single-data-point risk ratio of selection bias.

\subsection{Generalised Bayes and Reweighting}
A conventional Bayesian update for a model $p(X|\theta)$ trained on data points $X$ drawn from a true d.g.p. with density $g$ is known to converge on the point in parameter space $\theta$ that minimises the Kullback Leibler divergence between the d.g.p. and the statistical model:

\begin{align}
    KL(g(.)||p(.|\theta)) &= \int g(X) \log \frac{g(X)}{p(X|\theta)} dX \notag\\
    &= -  \int g(X) \log p(X|\theta) dX + c \notag\\
    &= \mathbb{E}_G\big[  -\ell_\theta(X) \big]+ c \notag\\
    & = \mathcal L(\theta;1)+ c
\end{align}

where $c$ is a constant that does not depend on $\theta$. In the final line above, the minimisation target is recast as the expectation of the negative log-likelihood $ \ell_\theta(X)$ of the statistical model under the data generating process $G$. This provides intuition as to why the distribution $G$ acts as the target of the inference, even when it differs from the statistical model. This can be estimated in the finite data sample using the empirical risk $\mathcal L_n(\theta)=-\frac{1}{n}\sum_{i=1}^n \,\ell_\theta(X_i)$, i.e. the sum over the negative log likelihood evaluated on the observed data under the model.

The key idea of our proposed method comes from generating a population-bias-robust posterior $\pi_n(\theta)$ by combining a prior $\pi(\theta)$ and  likelihood $p(X|\theta)$, with the contribution of each individual data point $X_i$ exponentiated with weights $w_i$:

\begin{equation}
    \pi_n(\theta) \propto \pi(\theta) \prod_{i=1}^n p(X_i|\theta)^{w_i}
\end{equation}

or equivalently, on the log scale:

\begin{align}
    \log \pi_n(\theta) &= \log \pi(\theta)  + \sum_{i=1}^n  w_i \ell_\theta(X_i) + c\notag\\
    &= \log \pi(\theta)  - n \mathcal L_n(\theta;w) + c 
\end{align}

The new reweighted empirical risk $\mathcal L_n(\theta;w) = \frac{1}{n}\sum_{i=1}^n w(X_i)\,(-\ell_\theta(X_i))$ is a function of the weights $w_i$: we propose a modular classifier-based computation of $r(X)=q(X)/g(X)$ evaluated at the observed data, reflecting the appropriate strength of contribution to the likelihood of each data point given its risk of selection bias, i.e.:
\begin{equation}
    \log \pi_n(\theta) = \log \pi(\theta)  + \sum_{i=1}^n  \frac{q(X_i)}{g(X_i)} \ell_\theta(X_i) + c
\end{equation}

This specific choice of weights for the empirical risk targets a population objective $\mathcal L(\theta;w)$ reweighted with $r(X) = \frac{q(X)}{g(X)}$ and hence the new KL:

\begin{align}\label{eqn:KL_swap}
\mathcal L(\theta;w) &= \mathbb{E}_G\big[r(X)\,(-\ell_\theta(X))\big]\notag\\
& = - \mathbb{E}_G\big[ \frac{q(X)}{g(X)} \ell_\theta(X)\big]\notag \\
& = - \int g(X) \frac{q(X)}{g(X)}  \ell_\theta(X) dX \notag\\
& = - \int q(X) \ell_\theta(X) dX\notag \\
& = - \mathbb{E}_Q\big[\ell_\theta(X)\big]\notag\\
& =  KL(q(.)||p(.|\theta))  + c
\end{align}

Consequently, reweighting by $w_i=q(X_i)/g(X_i)$ changes the population risk from one averaged under $G$ to one averaged under $Q$, thereby retargeting the corresponding population-level parameter toward $Q$. This requires the target distribution $Q$ to be absolutely continuous with respect to the source distribution $G$. The procedure can exploit observations from the larger sample drawn from $G$, however it does not generally provide the same information or uncertainty quantification as an equally large sample drawn directly from $Q$: its effective precision depends on the overlap between $G$ and $Q$ and on the variability of the importance weights.

Equation \ref{eqn:KL_swap} is valid when the weights $w_i$ are known and fixed, and more generally independent of $\theta$. However, the weights $w_i$ can be statistically or functionally dependent on $\theta$ through $g$, in which case a standard Bayesian workflow would prescribe a \emph{joint} model for $w_i$ and $\theta$. Yet this approach is vulnerable to what is known in the literature as the feedback problem \citep{mccandless2011cutting,frazier2025cutting}, and can produce estimators with suboptimal frequentist properties \citep{robins2015discussion}. To avoid this, we use a \emph{modular} procedure in which the weights are estimated first and then kept fixed \emph{as if they were known} before performing inference on $\theta$. This corresponds to a particular implementation of the ``Cut-Bayesian'' approach, used to modularize complex model specifications \citep{plummer2015cuts,liu2025general}.

\subsection{Effective Data Size (EDS)}

The question remains as to how powerful the update will be if the data points are reweighted to have varying degrees of influence: fortunately, we can use the related concept of the effective sample size from the Monte Carlo sampling literature, defined as $EDS = (\sum_{i=1}^n w_i)^2/\sum_{i=1}^n w_i^2$ \citep{kish1965survey}. This can be reintepreted at the data level: if the data generating distribution and bias-free distributions are equal, then a perfectly-trained classifier would identify all of the data points to be equally relevant with $w_i=1$ and we would recover an exact Bayesian update with an effective data size of $n$. However, if selection bias is severe in the sample, then the classifier may only select a small number of data points as being relevant for the update, and the effect data size will consequently shrink, as the update will be dominated by the small number of observed overlapping data points. Fortunately, this phenomenon will be identifiable from the computed weights and their derived statistics.

In Appendix \ref{app:classprag} we provide some further pragmatic advice on how to generate stable populations of weights for subsequent use in the belief update.

\section{Theoretical Properties}\label{sec:theory}

Here we provide some theoretical results with proofs under notation and assumptions in Appendices \ref{App:setup} and \ref{app:assump} respectively, standard for M-estimation analysis.

\begin{thm}[Existence and uniqueness]\label{thm:exist}
Let the following conditions be satisfied:
\begin{itemize}
    \item The parameter space $\Theta\subseteq\R^d$ is nonempty, compact and convex,
    \item $Q$ is absolutely continuous with respect to $G$ and $w:=dQ/dG$,
    \item $\mathcal L(\theta;w)=\mathbb E_G[w(X)\ell_\theta(X)]$ is finite, continuous, and strictly convex on $\Theta$.
\end{itemize}
Then there exists a unique target parameter $\theta^\star := \arg\min_{\theta\in\Theta}\mathcal L(\theta;w)$.
\end{thm}

\noindent\emph{Proof sketch.}\enspace Continuity of $\mathcal{L}(\cdot;w)$ on $\Theta$ follows from pointwise continuity of $\theta\mapsto\ell_\theta(x)$ and the integrable local envelope via dominated convergence. A minimum exists on every compact set by the extreme-value theorem; compactness of $\Theta$ provides global existence. Uniqueness under strict convexity holds because the midpoint of any two distinct minimisers would have strictly lower value, contradicting minimality.

\begin{remark}[General envelope conditions]\label{rem:exist-general}
The compactness assumption on $\Theta$ can be relaxed. If $\Theta$ is merely closed and $\mathcal{L}(\cdot;w)$ is coercive (sublevel sets compact), the same extreme-value argument yields existence; see Appendix~\ref{app:exist} for details. Strict convexity and convexity of $\Theta$ remain sufficient for uniqueness in this more general setting.
\end{remark}

See Appendix~\ref{app:exist} for the full proof.

\begin{thm}[Posterior concentration]\label{thm:consistency}
Assume \ref{A:param}, \ref{A:env}, \ref{A:GC}, \ref{A:ident},
and \ref{A:prior}. In addition, suppose that $\mathcal L(\cdot;w)$ is
finite on $\Theta$ and continuous at $\theta^\star$, and that the class
\[
    \mathcal F_\Theta
    =\{x\mapsto w(x)\ell_\theta(x):\theta\in\Theta\}
\]
is $G$-Glivenko--Cantelli. For every $\varepsilon>0$ for which
$A_\varepsilon=\{\theta\in\Theta:
\|\theta-\theta^\star\|\ge\varepsilon\}$ is nonempty, assume
\[
    \Delta_\varepsilon
    :=\inf_{\theta\in A_\varepsilon}
       \{\mathcal L(\theta;w)-\mathcal L(\theta^\star;w)\}>0.
\]
Then $\Pi_n(A_\varepsilon)\xrightarrow{\Pp}0$ for every $\varepsilon>0$.
\end{thm}

\noindent\emph{Proof sketch.}\enspace On the high-probability event $E_n$ where the empirical risk lies uniformly within $\Delta/8$ of the population risk (guaranteed by the Glivenko--Cantelli assumption), every $\theta\in A_\varepsilon$ has empirical risk at least $\mathcal{L}(\theta^\star;w)+7\Delta/8$, while every $\theta\in B_\delta$ has empirical risk at most $\mathcal{L}(\theta^\star;w)+\Delta/4$. The resulting exponential ratio of posterior masses decays as $\exp(-5n\Delta/8)/\Pi(B_\delta)$, which tends to zero since $\Pp(E_n)\to1$.

See Appendix~\ref{app:consistency} for the full proof.

\begin{thm}[Asymptotic normality and local posterior Gaussianity]\label{thm:asymp} Assume \ref{A:param}, \ref{A:env}, \ref{A:GC}, \ref{A:CLT}, \ref{A:ident}, and \ref{A:prior}. In addition, assume that the class $\{x\mapsto w(x)H_\theta(x):\theta\in N\}$ is $G$-Glivenko--Cantelli, so that \[ \sup_{\theta\in N} \|\mathbb{P}_n[w H_\theta] - \mathbb{E}_G[w H_\theta]\| \xrightarrow{P} 0. \] Let $\hat\theta_n$ be a sequence of estimators such that $\hat\theta_n\in N$ with probability tending to $1$ and \[ \nabla_\theta\mathcal L_n(\hat\theta_n;w) = 0, \qquad \hat\theta_n\xrightarrow{P}\theta^\star. \] Then \[ \sqrt{n}(\hat\theta_n-\theta^\star) \ \xrightarrow{d}\ \mathcal N(0,I^{-1}JI^{-1}), \] where \[ I = \mathbb{E}_G[w(X)H_{\theta^\star}(X)], \qquad J = \mathbb{E}_G[w(X)^2 s_{\theta^\star}(X)s_{\theta^\star}(X)^\top]. \] For the posterior conclusion, suppose additionally that all conditions of
Theorem~\ref{thm:consistency} hold. Then the generalised posterior satisfies, for any fixed $M>0$, \[ \sup_{B\subset\{u:\|u\|\le M\}} \left| \pi_n\big(\sqrt{n}(\theta-\hat\theta_n)\in B\big) - \mathcal N(0,I^{-1})(B) \right| \xrightarrow{P} 0.
\] \end{thm}

\noindent\emph{Proof sketch.}\enspace \textit{Part~(i).} A first-order Taylor expansion of $\nabla_\theta\mathcal{L}_n$ at $\theta^\star$ gives $\sqrt{n}(\hat\theta_n-\theta^\star)=-B_n^{-1}\sqrt{n}\nabla_\theta\mathcal{L}_n(\theta^\star;w)$, where $B_n\xrightarrow{\Pp}I$ by the Glivenko--Cantelli condition on the Hessian class and consistency of $\hat\theta_n$. The score CLT (Assumption~\ref{A:CLT}) and Slutsky's lemma then give the sandwich distribution $\mathcal{N}(0,I^{-1}JI^{-1})$. \textit{Part~(ii).} A second Taylor expansion of the log-posterior around $\hat\theta_n$ shows it is locally proportional to $\exp(-\frac{n}{2}u^\top Iu)$, confirming the $\mathcal{N}(0,I^{-1})$ shape; tail control uses the exponential lower bound on the normalising constant inherited from the consistency proof.

See Appendix~\ref{app:asymp} for the full proof.

\begin{remark}[Posterior vs.\ frequentist variance]\label{rem:variance}
The Gaussian posterior approximation has covariance $I^{-1}/n$, whereas
the sampling distribution of $\hat\theta_n$ has asymptotic covariance
$I^{-1}JI^{-1}/n$. These matrices coincide if and only if $J=I$;
correct specification with $w\equiv1$ is a sufficient example.
In general the posterior is miscalibrated, and neither covariance matrix
need dominate the other. Frequentist calibration can use the sandwich
covariance $I^{-1}JI^{-1}/n$ in a Gaussian approximation centred at
$\hat\theta_n$. The covariance of the approximating Gaussian is being
described here; convergence of actual posterior second moments requires
additional moment or tail conditions. We investigate the calibration shortfall and a sandwich-based correction empirically in Section~\ref{sec:sandwich}.
\end{remark}

\begin{remark}[Concrete sufficient conditions]\label{rem:concrete}
In many applications the abstract Glivenko--Cantelli and CLT conditions in Theorems~\ref{thm:consistency} and~\ref{thm:asymp} are implied by a single smoothness assumption. Suppose $\Theta$ is compact and $\theta\mapsto w(x)\ell_\theta(x)$ is three times continuously differentiable on an open neighbourhood of $\Theta$ for $G$-almost every $x$, with integrable third-derivative envelope:
\[
    \mathbb E_G\bigl[\sup_{\theta\in\Theta}\|\nabla^3_\theta w(X)\ell_\theta(X)\|\bigr]<\infty.
\]
Then the loss class and the Hessian class are each $G$-Glivenko--Cantelli by a standard covering-number argument. The score CLT (Assumption~\ref{A:CLT}) additionally requires $\mathbb E_G[w(X)^2\|s_{\theta^\star}(X)\|^2]<\infty$. Both the 1D Gaussian and the 4D linear-model examples in Section~\ref{sec:results} satisfy these conditions with easily verified moment bounds.
\end{remark}

\begin{thm}[Asymptotic effect of plug-in weights]\label{thm:plugin} Assume \ref{A:param}--\ref{A:CLT}, \ref{A:prior}, and \ref{A:ratio}, and the Glivenko--Cantelli condition for $\{w H_\theta:\theta\in N\}$ as in Theorem \ref{thm:asymp}. Suppose $\hat\gamma_m$ is based on an independent sample of size $m$, with \[ \sqrt{m}(\hat\gamma_m-\gamma_0) \xrightarrow{d} \mathcal N(0,\Sigma_\gamma), \] and define $\hat r_n(x) = r_{\hat\gamma_m}(x)$. Let $\tilde\theta_n$ minimise $\mathcal L_n(\theta;\hat r_n)$ or solve $\nabla_\theta\mathcal L_n(\tilde\theta_n;\hat r_n)=0$, and assume $\tilde\theta_n\xrightarrow{P}\theta^\star$. Let \[
    A:=\mathbb E_G[s_{\theta^\star}(X)\dot r_{\gamma_0}(X)^\top]
       \in\mathbb R^{d\times p},\qquad
    c:=\lim_{n\to\infty}\frac{n}{m}\in[0,\infty),\qquad
    K_c:=cA\Sigma_\gamma A^\top.
\] where $c := \lim_{n\to\infty} m/n\in[0,\infty]$ (if the limit exists). \begin{enumerate} \item If $m/n\to\infty$ (i.e.\ $c=0$), then \[ \sqrt{n}(\tilde\theta_n-\theta^\star) \xrightarrow{d} \mathcal N(0,I^{-1} J I^{-1}). \] \item If $n/m\to c\in(0,\infty)$, then \[ \sqrt{n}(\tilde\theta_n-\theta^\star) \xrightarrow{d} \mathcal N\big(0, I^{-1}(J+K_c)I^{-1}\big). \] \end{enumerate} \end{thm}

\noindent\emph{Proof sketch.}\enspace A first-order Taylor expansion of $r_{\hat\gamma_m}$ around $\gamma_0$ transfers the weight-estimation error $\sqrt{m}(\hat\gamma_m-\gamma_0)\xrightarrow{d}\mathcal{N}(0,\Sigma_\gamma)$ into an additional term $-A_n\sqrt{n}\,\delta_m$ in the score equation, where $A_n\xrightarrow{\Pp}A$ by the law of large numbers. Independence of the auxiliary sample and the oracle score CLT then yield a joint Gaussian limit, with the additional variance $K_c=cA\Sigma_\gamma A^\top$ arising from ratio-estimation uncertainty; this vanishes when $n/m\to0$.

See Appendix~\ref{app:plugin} for the full proof.

\begin{remark}
The case $m/n\to\infty$ ($c=0$, $K_c=0$) recovers the oracle sandwich variance
$I^{-1}JI^{-1}$ of Theorem~\ref{thm:asymp}: when the weight-estimating sample is
much larger than the inference sample, plug-in weight estimation contributes
negligible additional uncertainty.  Conversely, when $m\asymp n$ ($c\in(0,\infty)$),
the variance inflation $K_c=cA\Sigma_\gamma A^\top$ is non-negligible and must be
accounted for in inference. The regime $m/n\to0$ is outside this root-$n$ result. In that regime the auxiliary estimator may contribute the leading sampling uncertainty;a separate rate analysis is required. Here $m$ denotes the auxiliary
training-sample size under the stated independent-sample asymptotics.
\end{remark}

We also provide conditions on the distributions $G$ and $Q$ allowing for successful practical ratio estimation in Appendix \ref{app:gq}.

\section{Results}\label{sec:results}

\subsection{Gaussian $G$ and $Q$}

Here we present a simple example performing inference for the mean and variance of a Gaussian statistical model $p(X|\mu,\sigma) = \mathcal{N}(\mu,\sigma^2)$, when the observed data $G\sim \mathcal{N}(0,2)$ exhibits a selection bias relative to the ideal study population of $Q\sim\mathcal{N}(1,1)$. The standard Bayes posterior trained on the observed data $X\sim G$ would thus exhibit bias for the mean $\mu$ and standard deviation $\sigma$.

We demonstrate the empirical behaviour of the Bayesian reweighting schemes for this simple example. Non-informative priors of $\mathcal{N}(0,5)$ were used for both $p(\mu)$ and $p(\sigma)$, and are not expected to contribute significantly to the posterior.

200 ``low quality'' data points X were observed to be drawn from the biased distribution $G$ and 50 ``high quality'' data points $X'$ were available from the ideal distribution $Q$. Logistic regression trained on an orthogonal basis up to order 4 was used as a classifier to estimate the ratios for the weights $w_i$. Weights were estimated using 5-fold cross-validation, and were computed ahead of time before the update inference was performed in brms \citep{burkner2017brms}, an R package for specifying glm regression models, which interfaces with Stan for fully Bayesian HMC inference in R.

For comparison, we include results for five different approaches:

\begin{enumerate}
    \item Standard Bayes trained on the observed data $X$ without correction: $Standard$.
    \item Standard Bayes trained on target population data $X'$ without correction: $Target$-$only$.
    \item Bayes weighted according to the Oracle-estimated ratio $r(x)=q(X)/g(X)$, based on the true generate distribution, targeting the exact idealised distribution $q(X)$: $Odds$.
    \item Bayes weighted according to the classifier-estimated ratio $r(x)=q(X)/g(X)$ using the observed data, targeting the exact idealised distribution $q(X)$: $Oracle$.
    \item As described in Appendix \ref{app:classprag}, Bayes reweighted by the tempered classifier-estimated ratio $r(X)^{\alpha} = q(X)^{\alpha}/g(X)^{\alpha}$, using the observed data, targeting the ratio distribution $q(X)^{\alpha}g(X)^{1-\alpha}$ with $\alpha = 0.75$, expected to show greater stability than the $\alpha=1$ updates: $Odds$ $\alpha$.
\end{enumerate}

We derived systematic results averaged over 100 random seeds. The posterior mean estimates and coverage for each parameter by each method are plotted in Figure \ref{fig1}. Here, coverage refers to the empirical frequentist coverage of the raw 90\% generalised-posterior credible intervals. As discussed in Remark~\ref{rem:variance}, these intervals are not generally expected to attain nominal frequentist coverage under reweighting unless an appropriate variance calibration is applied.

We also present the weights and log weights for a single random seed, and the EDS for all random seeds in Figure \ref{fig2}. The seed-averaged posterior biases, RMSEs, SDs, coverages and EDS are presented numerically in Table \ref{tab1}.

We observe that, as expected, the oracle and target-population estimates consistently target the true parameter value, while the standard Bayes model trained on the observed data $X$ exhibits a clear bias. The $Odds$ method also targets the true value relatively consistently, with a small bias in the direction of the observed data $X$. The $Odds$ $\alpha$ exhibits a larger bias towards the observed data $X$ and worse coverage compared to $Odds$, but with the advantage of a larger EDS and hence more stable SD and RMSE in some cases. This corresponds to our expectations that the $Odds$ reweighting will target the exact distribution, but in a less stable manner compared to the inexact $Odds$ $\alpha$.

\begin{figure}
\begin{subfigure}{0.99\textwidth}
  \centering
  \includegraphics[width=.99\linewidth]{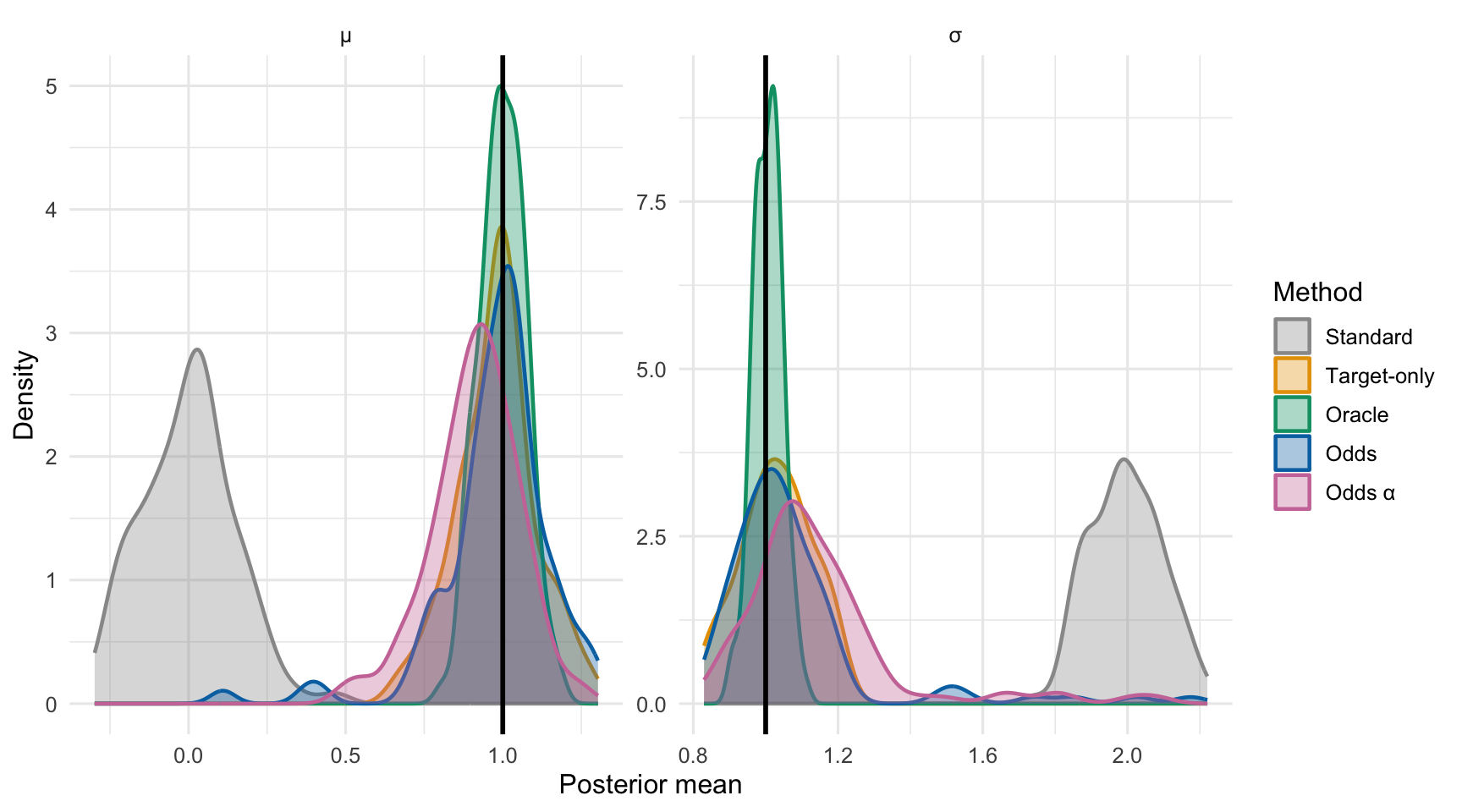}
  \caption{Posterior mean distributions for a Gaussian distribution with varying reweighting methods.}
  \label{fig1:sfig1}
\end{subfigure}
\begin{subfigure}{0.99\textwidth}
  \centering
  \includegraphics[width=.99\linewidth]{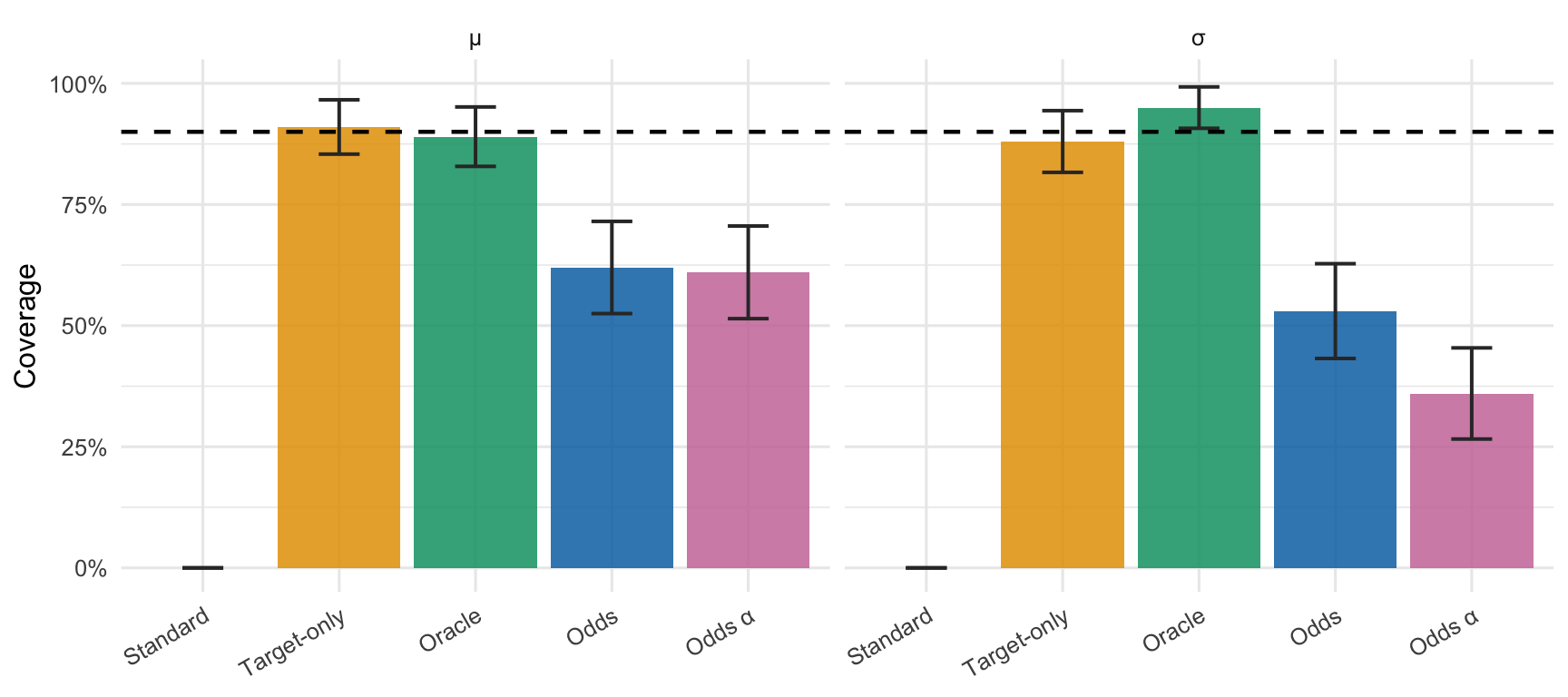}
  \caption{Coverages for a Gaussian distribution with varying reweighting methods.}
  \label{fig1:sfig2}
\end{subfigure}
\caption{Model parameter results comparing different reweighting methods for inference of a Gaussian distribution with unknown mean and variance.}
\label{fig1}
\end{figure}

\begin{figure}
\begin{subfigure}{0.99\textwidth}
  \centering
  \includegraphics[width=.99\linewidth]{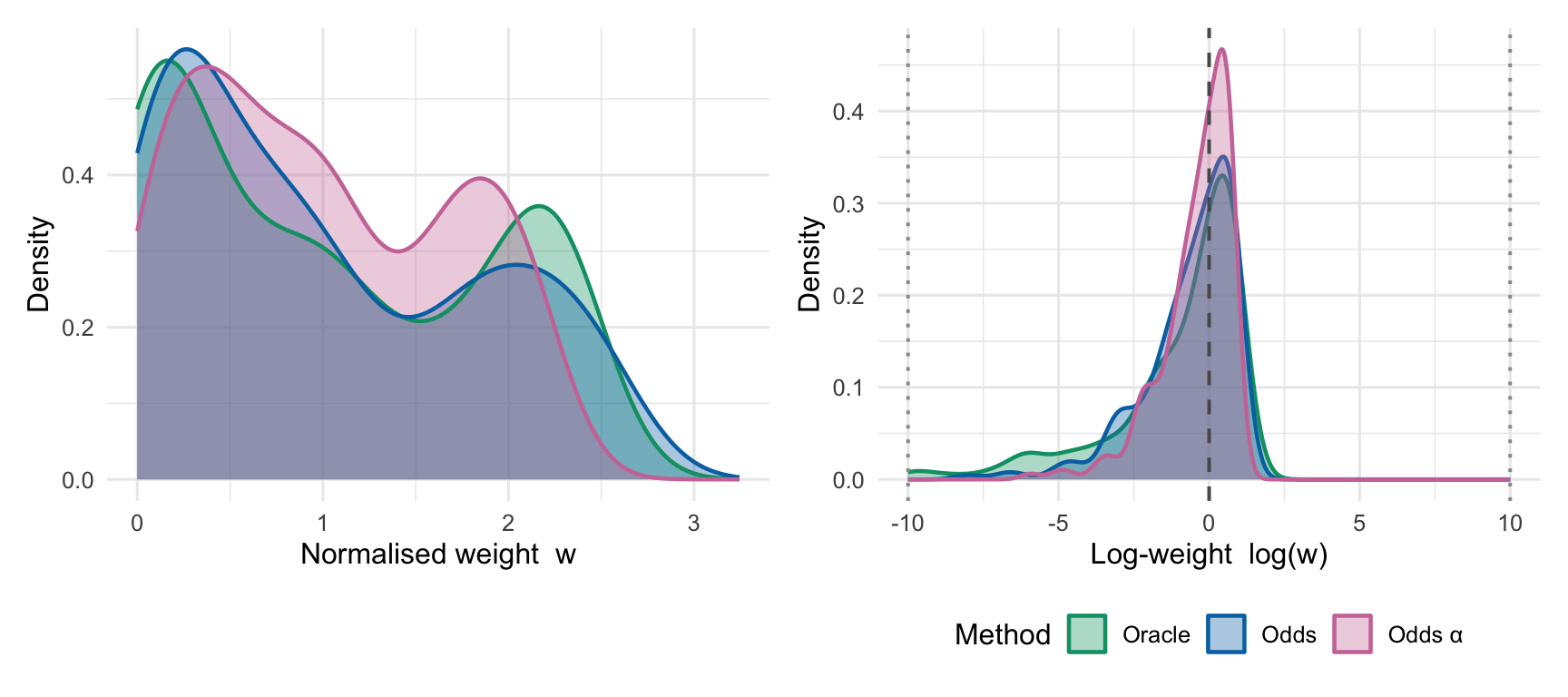}
  \caption{Weights with and without a log transformation for a single random seed under different weighting methods.}
  \label{fig2:sfig1}
\end{subfigure}
\begin{subfigure}{0.99\textwidth}
  \centering
  \includegraphics[width=.99\linewidth]{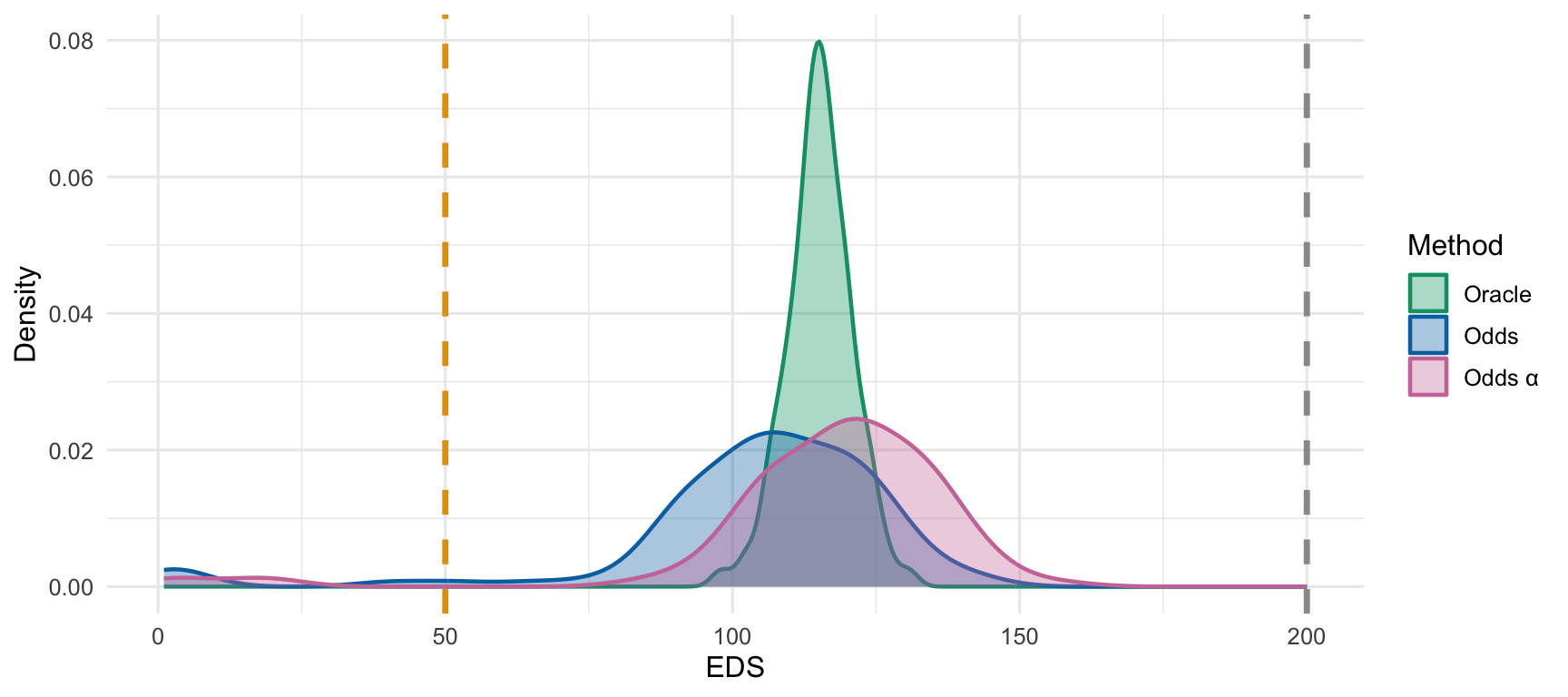}
  \caption{Effective Data Size (EDS) for a Gaussian distribution varying reweighting methods.}
  \label{fig2:sfig2}
\end{subfigure}
\caption{Model diagnostics comparing different reweighting methods for inference of a Gaussian distribution with unknown mean and variance.}
\label{fig2}
\end{figure}

\begin{table}[ht]
\centering
\caption{Performance metrics by method for parameters $\mu$ and $\sigma$ averaged over 100 seeds. Bias, RMSE, and SD(mean) are relative to the true generative parameters of $Q$ ($\mu^*=1$, $\sigma^*=1$). Coverage is for the 90\% posterior credible interval. \textit{Oracle} uses the analytically known density ratio; \textit{Target-only} fits standard Bayes on the small unbiased sample only.}
\label{tab1}
\begin{tabular}{lrrrrr}
\toprule
Method & Bias & RMSE & SD(mean) & Coverage & EDS \\
\midrule
\multicolumn{6}{l}{\textbf{Parameter: $\mu$}} \\
\midrule
Standard & $-1.005$ & $1.016$ & $0.149$ & $0.00$ & $200$ \\
Target-only & $-0.016$ & $0.129$ & $0.129$ & $0.91$ & $50$ \\
Oracle & $-0.001$ & $0.072$ & $0.072$ & $0.89$ & $115$ \\
Odds & $-0.100$ & $1.107$ & $1.108$ & $0.62$ & $103$ \\
Odds $\alpha$ & $-0.162$ & $0.953$ & $0.944$ & $0.61$ & $116$ \\
\midrule
\multicolumn{6}{l}{\textbf{Parameter: $\sigma$}} \\
\midrule
Standard & $0.998$ & $1.003$ & $0.101$ & $0.00$ & $200$ \\
Target-only & $0.028$ & $0.105$ & $0.102$ & $0.88$ & $50$ \\
Oracle & $0.006$ & $0.040$ & $0.040$ & $0.95$ & $115$ \\
Odds & $0.103$ & $0.360$ & $0.346$ & $0.53$ & $103$ \\
Odds $\alpha$ & $0.179$ & $0.366$ & $0.321$ & $0.36$ & $116$ \\
\bottomrule
\end{tabular}
\end{table}

\subsection{Multivariable Regression}

In this section we consider a more complex and realistic example of registry data usage, involving estimating the influence of various covariates on outcome of interest in a registry data set subject to selection bias. We consider the outcome as a linear combination of a measured biomarker, age, sex, and an unmeasured noise variable, i.e. $outcome = biomarker + age + sex + noise$, with $biomarker$, $age$, and $noise$ being drawn from $\mathcal{N}(0,1)$, and $sex$ drawn from $\text{Bernoulli}(p=0.5).$

The target population consists of 200 samples from the idealised distribution  $Q$, while 500 samples from the registry distribution $G$ are generated with a biasing mechanism, in which data points are sampled without replacement from a population of 1000 with probability $p=(1+exp(outcome - median(outcome)))^{-1}$, e.g. patients with higher outcome scores are more likely to be excluded from the registry population. This represents a bias that will realistically exist in observational data, as more severe cases are harder to include in registries.

The weights $w_i$ were estimated using 5-fold CV and logistic regression, trained on all  quadratic moments and pairwise combinations of the observed variables. These were precomputed before the belief update, which was performed using brms to estimate regression coefficients for each covariate, plus an intercept and noise sd $\sigma$.

We derived results over 100 random seeds to capture the stochastic behaviour systematically. We present posterior means and EDS across seeds in Figure \ref{fig3}. We further present posterior coverages across seeds, plus an example distribution of weights for a single random seed in Figure \ref{fig4}. Average posterior biases, RMSEs, SDs, coverages and EDS values are presented numerically in Table \ref{tab2}.

As expected, we observe that Standard uncorrected Bayes performs the worst, while Target-only performs the best. We do not see a consistent ordering between Oracle, Odds, and Odds $\alpha$ in bias, although all of them tend to be biased towards the Standard estimator. Similarly, none of those three consistently have superior coverage for all parameters. Odds $\alpha$ does have a larger EDS compared to Oracle and Odds, resulting in  smaller mean SDs and often (but not universally) smaller RMSEs. These results again meet our expectations that Odds targets the exact distribution, but in a less stable way than the more conservative Odds $\alpha$.

\subsection{Sandwich Variance Calibration}
\label{sec:sandwich}

It is clear from Remark~\ref{rem:variance} that generalised posteriors will have miscalibrated coverages under non-trivial reweighting. Specifically, for $J\neq I$, they will have covariance $I^{-1}/n$ instead of the inflated sandwich covariance $I^{-1}JI^{-1}/n$, which better quantifies the uncertainty of the estimator. Fortunately, the sandwich covariance can be estimated from the data and applied as a post-processing correction of the reported model variances: we pursue this approach for the previous two simulated examples in this section.

The correction replaces the raw posterior half-widths by a factor $K_j = \sqrt{[I^{-1}JI^{-1}]_{jj}\,/\,[I^{-1}]_{jj}}$ for each scalar parameter $j$, where $I$ and $J$ are estimated as
\begin{align}
  \hat{I} &= -\frac{1}{n}\sum_{i=1}^n w_i \nabla^2_\theta \log p(x_i \mid \hat\theta), \label{eq:Ihat}\\
  \hat{J} &= \frac{1}{n}\sum_{i=1}^n w_i^2\,\nabla_\theta\log p(x_i\mid\hat\theta)\,\nabla_\theta\log p(x_i\mid\hat\theta)^\top, \label{eq:Jhat}
\end{align}
using the posterior mean $\hat\theta$ as a plug-in estimator.  The corrected 90\% interval for parameter $j$ is then
\[
  \bigl[\hat\theta_j + K_j(\hat\theta^{-}_{j,0.9} - \hat\theta_j),\;
        \hat\theta_j + K_j(\hat\theta^{+}_{j,0.9} - \hat\theta_j)\bigr],
\]
where $\hat\theta^{\pm}_{j,0.9}$ are the raw 5th and 95th posterior quantiles. The computation requires only the data, the weights, and the posterior mean, with no additional model fits.

Table~\ref{tab:sandwich} reports raw and corrected empirical 90\% coverages together with the mean correction factor $K$, for each IPW method across both simulation experiments, averaged over 100 seeds.  Standard and Target-only are omitted, as neither uses reweighing that would generate associated sandwich inflation.

Several patterns emerge from Table~\ref{tab:sandwich}.  Firstly, $K>1$ for all IPW methods in both experiments, empirically confirming that the weighted posterior systematically underestimates frequentist variance, as expected. Secondly, the correction moves coverage toward the nominal 0.90 in most cases, most markedly for the Oracle method in the 4D experiment.  Thirdly a residual gap below 0.90 remains for all classifier-estimated-weight methods (Odds and Odds $\alpha$): this is consistent with the additional variance inflation arising from plug-in weight variability (Theorem~\ref{thm:plugin}), which is not captured by the fixed-weight sandwich inflation. Finally, for parameters where undercoverage is primarily attributable to estimator bias rather than variance, the correction offers no benefit, since widening a systematically shifted interval does not recover coverage: this is most notable $\sigma$ in the 1D experiment under Odds and Odds $\alpha$, which are structurally most vulnerable to finite-sample effects from individual data points in the tails. The results empirically validate the sandwich structure of Theorem~\ref{thm:asymp}.

\begin{figure}
  \centering
  \includegraphics[width=.9\linewidth]{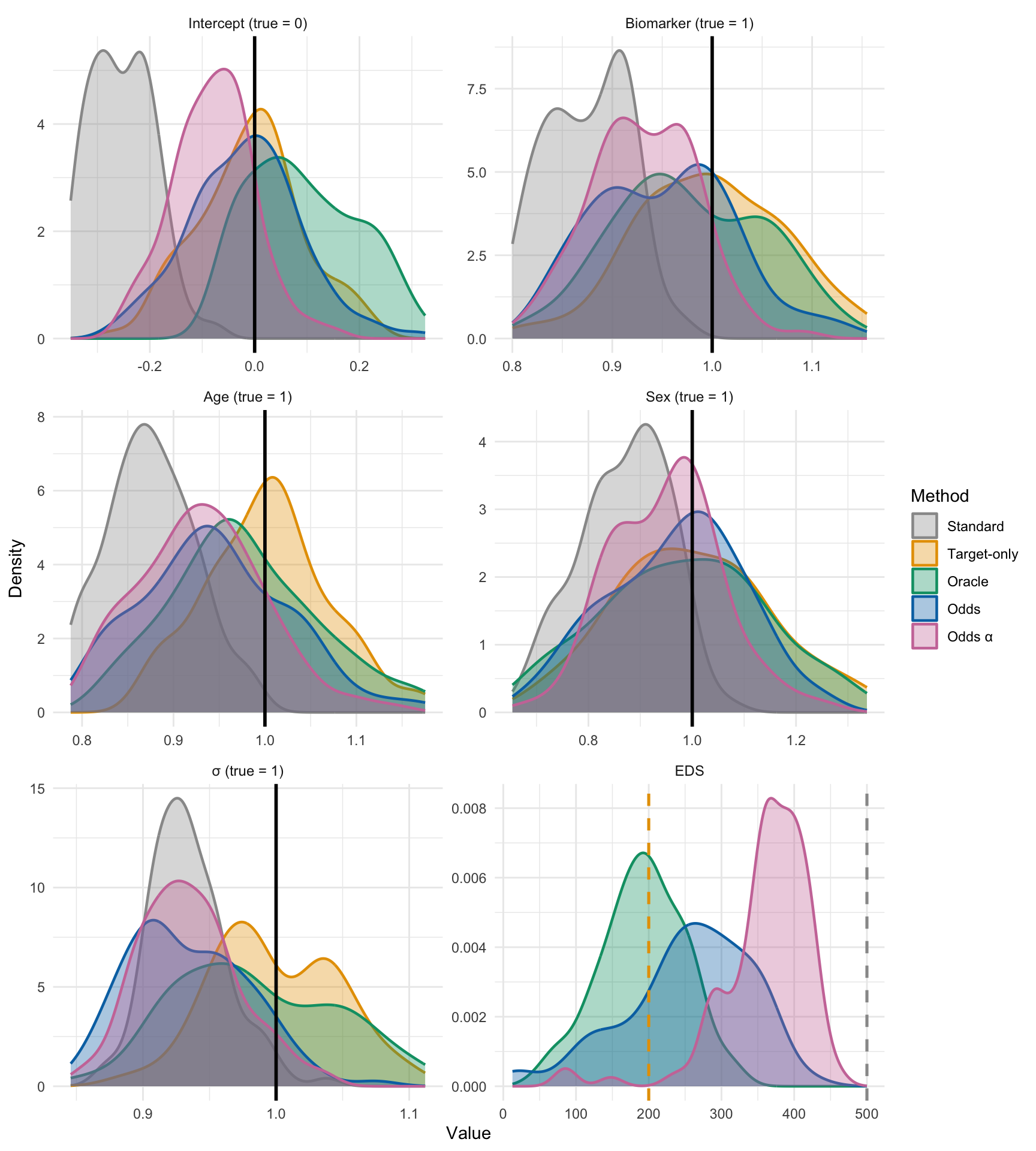}
\caption{Posterior means and EDS values comparing different reweighting methods for a simulated epidemiological model using multivariable linear regression.}
\label{fig3}
\end{figure}

\begin{figure}
\begin{subfigure}{0.99\textwidth}
  \centering
  \includegraphics[width=.9\linewidth]{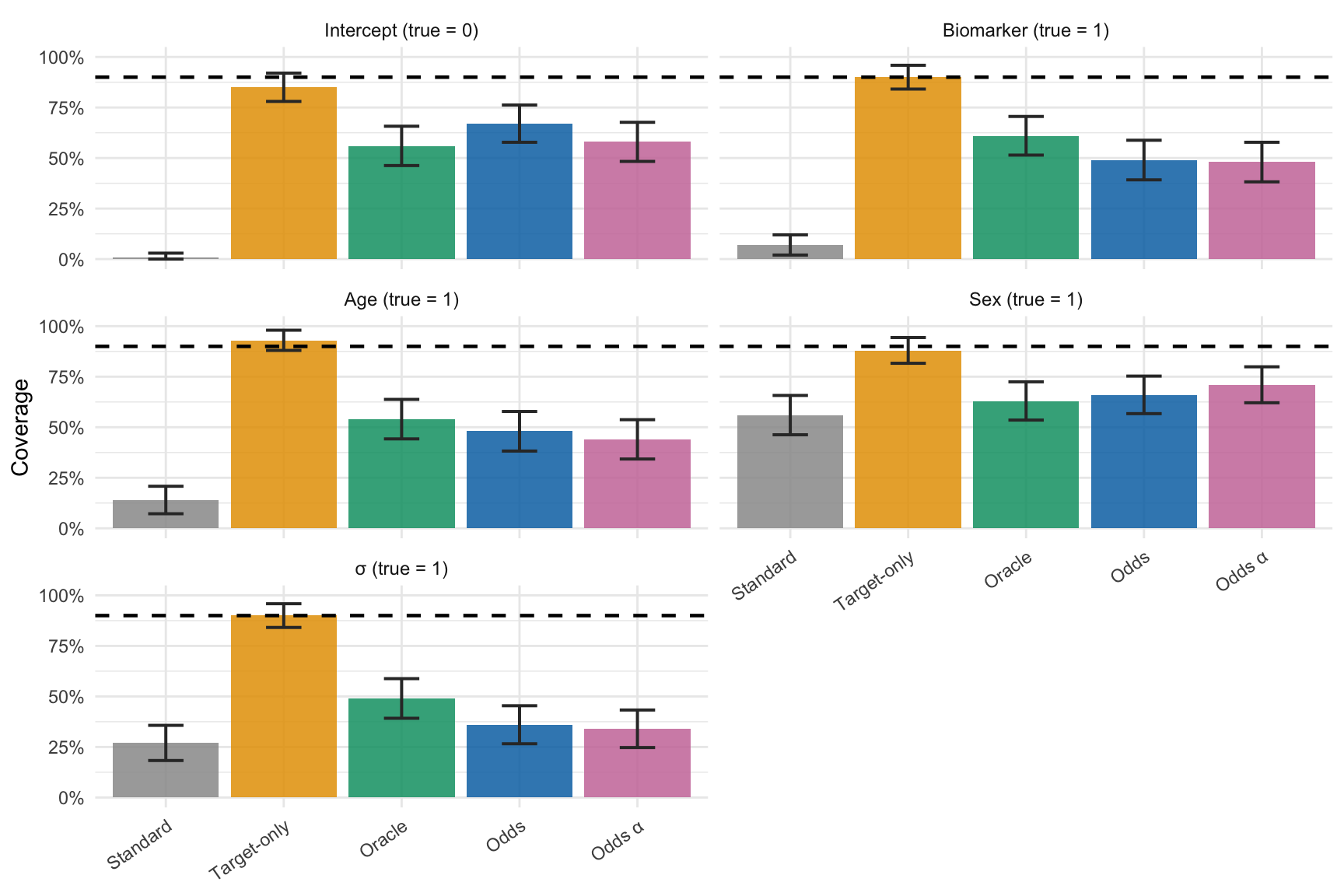}
  \caption{Posterior distributions for the intercept parameter of a linear regression under different reweighting methods.}
  \label{fig4:sfig1}
\end{subfigure}
\begin{subfigure}{0.99\textwidth}
  \centering
  \includegraphics[width=.9\linewidth]{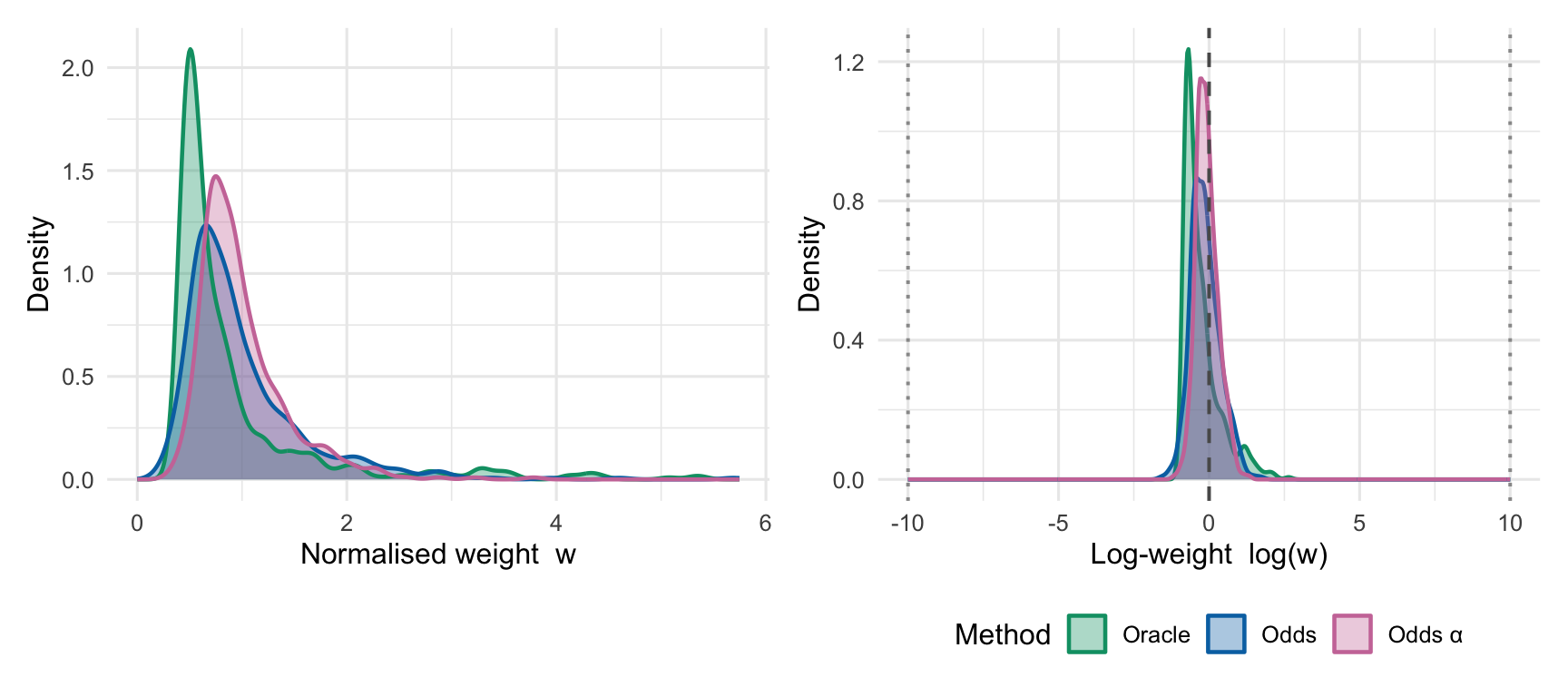}
  \caption{Posterior distributions for the biomarker coefficient of a linear regression under different reweighting methods.}
  \label{fig4:sfig2}
\end{subfigure}
\caption{Example inference results comparing different reweighting methods for a simulated epidemiological model using multivariable linear regression.}
\label{fig4}
\end{figure}

\begin{table}[ht]
\centering
\caption{Performance metrics by method for the 4D regression simulation averaged over 100 seeds. Bias, RMSE, and SD(mean) are relative to the true generative parameters (all regression coefficients $=1$, intercept $=0$, $\sigma=1$). Coverage is for the 90\% posterior credible interval. \textit{Oracle} uses inverse selection probabilities $1/p_i$; \textit{Target-only} fits standard Bayes on the $n=200$ unbiased sample only.}
\label{tab2}
\begin{tabular}{lrrrrr}
\toprule
Method & Bias & RMSE & SD(mean) & Coverage & EDS \\
\midrule
\multicolumn{6}{l}{\textbf{Parameter: Intercept}} \\
\midrule
Standard & $-0.260$ & $0.268$ & $0.067$ & $0.01$ & $500$ \\
Target-only & $-0.004$ & $0.106$ & $0.107$ & $0.85$ & $200$ \\
Oracle & $0.103$ & $0.152$ & $0.113$ & $0.56$ & $192$ \\
Odds & $-0.018$ & $0.107$ & $0.106$ & $0.67$ & $255$ \\
Odds $\alpha$ & $-0.079$ & $0.110$ & $0.077$ & $0.58$ & $360$ \\
\midrule
\multicolumn{6}{l}{\textbf{Parameter: Biomarker}} \\
\midrule
Standard & $-0.127$ & $0.134$ & $0.044$ & $0.07$ & $500$ \\
Target-only & $-0.003$ & $0.073$ & $0.073$ & $0.90$ & $200$ \\
Oracle & $-0.018$ & $0.088$ & $0.086$ & $0.61$ & $192$ \\
Odds & $-0.046$ & $0.091$ & $0.079$ & $0.49$ & $255$ \\
Odds $\alpha$ & $-0.067$ & $0.088$ & $0.057$ & $0.48$ & $360$ \\
\midrule
\multicolumn{6}{l}{\textbf{Parameter: Age}} \\
\midrule
Standard & $-0.126$ & $0.135$ & $0.049$ & $0.14$ & $500$ \\
Target-only & $0.002$ & $0.066$ & $0.067$ & $0.93$ & $200$ \\
Oracle & $-0.018$ & $0.093$ & $0.091$ & $0.54$ & $192$ \\
Odds & $-0.053$ & $0.105$ & $0.091$ & $0.48$ & $255$ \\
Odds $\alpha$ & $-0.071$ & $0.100$ & $0.070$ & $0.44$ & $360$ \\
\midrule
\multicolumn{6}{l}{\textbf{Parameter: Sex}} \\
\midrule
Standard & $-0.130$ & $0.156$ & $0.086$ & $0.56$ & $500$ \\
Target-only & $0.000$ & $0.147$ & $0.148$ & $0.88$ & $200$ \\
Oracle & $-0.020$ & $0.171$ & $0.171$ & $0.63$ & $192$ \\
Odds & $-0.025$ & $0.158$ & $0.156$ & $0.66$ & $255$ \\
Odds $\alpha$ & $-0.049$ & $0.125$ & $0.116$ & $0.71$ & $360$ \\
\midrule
\multicolumn{6}{l}{\textbf{Parameter: $\sigma$}} \\
\midrule
Standard & $-0.066$ & $0.072$ & $0.029$ & $0.27$ & $500$ \\
Target-only & $0.004$ & $0.053$ & $0.053$ & $0.90$ & $200$ \\
Oracle & $-0.013$ & $0.064$ & $0.063$ & $0.49$ & $192$ \\
Odds & $-0.071$ & $0.087$ & $0.050$ & $0.36$ & $255$ \\
Odds $\alpha$ & $-0.067$ & $0.077$ & $0.038$ & $0.34$ & $360$ \\
\bottomrule
\end{tabular}
\end{table}

\subsection{Real Data Example}

Here we present a real data example to demonstrate the use of this method in a realistic large-scale registry data example. We take the example of using Prostate Specific Antigen (PSA) to predict 10-year mortality from prostate cancer (PCa) after diagnosis. We use the Norwegian Prostate Cancer Consortium data set, previously published in \citep{oldenburg2022long}. We pursue this example to demonstrate the behaviour of the proposed inference method on a real data set for a simple predictive model: we do not intend this to be a definitive contribution of rigorous epidemiological analysis to the somewhat contentious ongoing academic debate concerning PSA screening.

We consider the model specification of a univariate logistic regression model, using log-transformed most-recent pre-PCa-diagnosis PSA measurement to predict death from PCa within ten years. It would be of interest to understand the predictive ability of the widely used biomarker on an important clinical outcome. It is understood that the predictive ability varies with age: PSA naturally rises with age, so an elevated PSA at a lower age represents a more severe concern than an equivalent value at an older age. However, the data set represents the wide heterogeneity in PSA screening practices across all of Norway historically, such that it is hard to interpret models trained naively using the entire data set. Several age-related effects are relevant: 

\begin{itemize}
    \item The true underlying distribution of PCa diagnoses is more common in older men, with an older age distribution compared to the general population.
    \item PSA screening is targeted at older men in the population, with the result that PCa cases in younger men are frequently undetected (often harmlessly). The distribution of ages at diagnosis is thus higher than the true unknown distribution of PCa onset in the general population.
    \item Elevated PSA values at a lower age represent an elevated risk of death from PCa compared to equivalent PSA values at an older age.
    \item We acknowledge several other effects, including the competing risks of death from other causes, tumour stage at diagnosis, and frequency of testing. We would reserve a detailed treatment of these effects for an epidemiology article, but we do not expect them to reverse the general trend that elevated PSA at PCa diagnosis is a stronger predictor of 10-year death from PCa in a younger population compared to an older population.
\end{itemize}

There are four distributions of interest:

\begin{itemize}
    \item $G$ with density $g(X)$, the distribution of individuals with a PCa diagnosis in the observed data set
    \item $Q$ with density $q(X)$, the distribution of all men 39 or older in the Norwegian population
    \item $P$ with density $p(X)$, the underlying, unobserved distribution of PCa onset in the population
    \item $H$ with density $h(X)$, the artificially biased subsample of the observed data used for model testing
\end{itemize}

Demographic information on the distribution of ages among men in the general Norwegian population was found on the Statistics Norway website and used to construct comparison populations \citep{ssb2026table}. Histograms of example populations used are presented in Figure \ref{fig5:pop_histograms}.

Estimates of the influence of log PSA on 10-year PCa mortality when trained on data from $G$ will have a negative bias relative to the true distribution $P$ as the data set systematically misses younger men for whom elevated PSA is a greater concern. We are unable to use the reweighting strategy to target the true distribution $P$ as it is unobserved: instead, we propose reweighting to the distribution $Q$, as the general population would be expected to have a lower age distribution than the distribution $P$: this analysis would offer a principled upper bound on the influence of log PSA on PCa mortality to complement the lower bound provided by the observed registry data. Effectively, we consider a hypothetical population in which every Norwegian man 39 or older is equally likely to be diagnosed with PCa. Since we know that PCa is an illness that primarily impacts older men (and vanishingly few under 39), we can therefore assume that this age distribution $Q$ is downwardly biased relative to the true unobserved population $P$ and therefore exhibits a positive bias to the observed data from distribution $G$.

We considered two data paradigm comparisons to investigate the properties of the method. The first verifies the robustness of the reweighting process to target out-of-sample age distributions, while the second represents the projection onto the unknown full population distribution $Q$.

\begin{figure}
    \centering
    \includegraphics[width=0.9\linewidth]{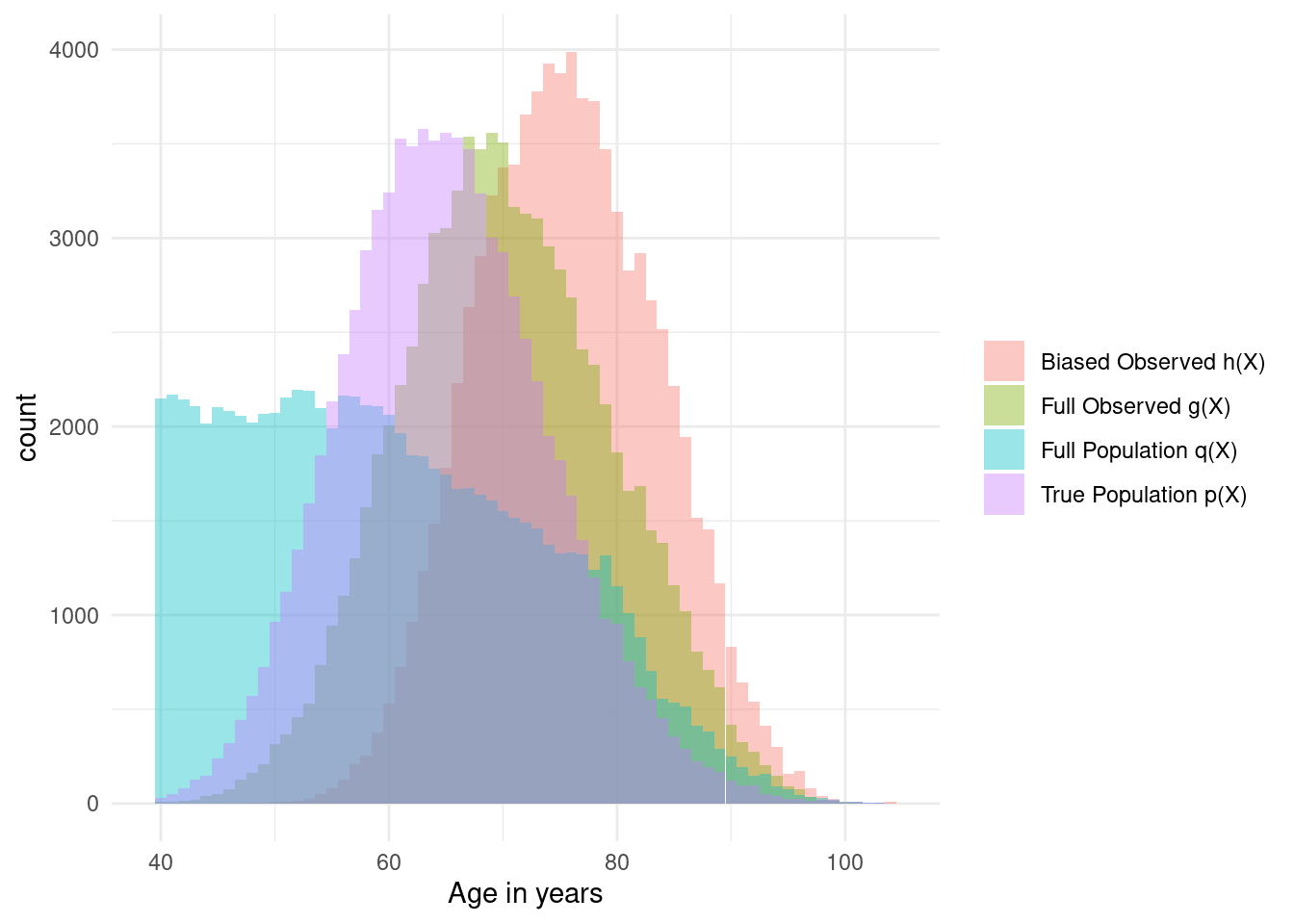}
    \caption{Example histograms of populations used in the PCa example. Biased Observed is $h(X)$, the deliberately biased observed population for model testing, Full Observed is $g(X)$, the entire observed registry population, True Population is $p(X)$, the unobserved true age of onset (presented schematically here), and Full Population is $q(X)$, the age distribution of all men aged 39 or over in Norway. Their central tendencies decrease in that order.}
    \label{fig5:pop_histograms}
\end{figure}

\begin{figure}
\begin{subfigure}{0.99\textwidth}
  \centering
  \includegraphics[width=.9\linewidth]{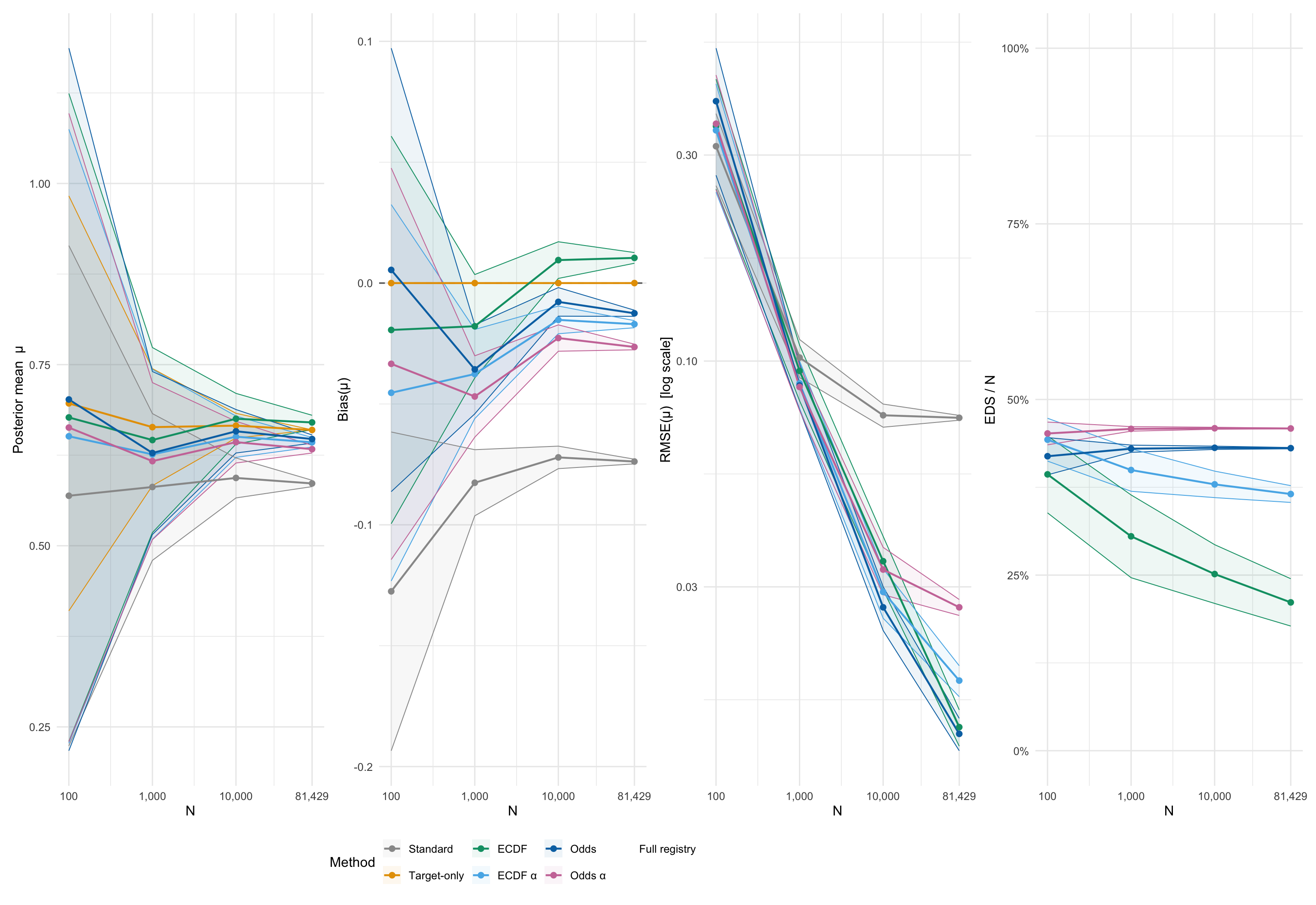}
  \caption{Posterior statistics for the $\mu$ parameter and EDS.}
  \label{fig6:sfig1}
\end{subfigure}
\begin{subfigure}{0.99\textwidth}
  \centering
  \includegraphics[width=.9\linewidth]{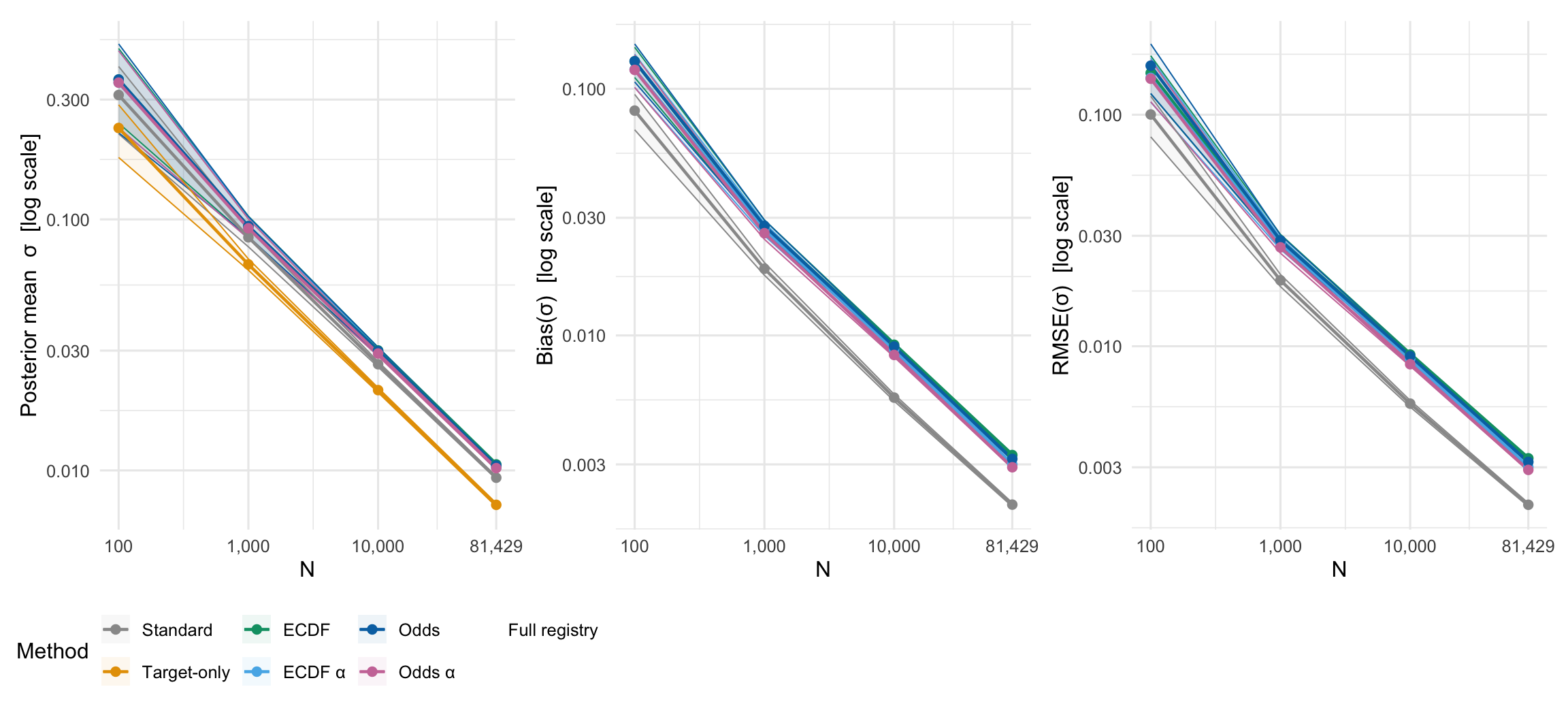}
  \caption{Posterior statistics for the $\sigma$ parameter.}
  \label{fig6:sfig2}
\end{subfigure}
\caption{Results for the out-of-sample analysis trained on artificially biased registry data from $h(X)$, predicting the observed registry distribution from $g(X)$. Statistics reported are means, biases and RMSEs for posterior means ($\mu$) and SDs ($\sigma$), plus the EDS. Results are averaged over random seeds with confidence intervals, representing  different methods with different data sample sizes N. RMSE and bias values are calculated relative to the Full registry posterior trained on the target data population.}
\label{fig6}
\end{figure}

\begin{figure}
    \centering
    \includegraphics[width=0.9\linewidth]{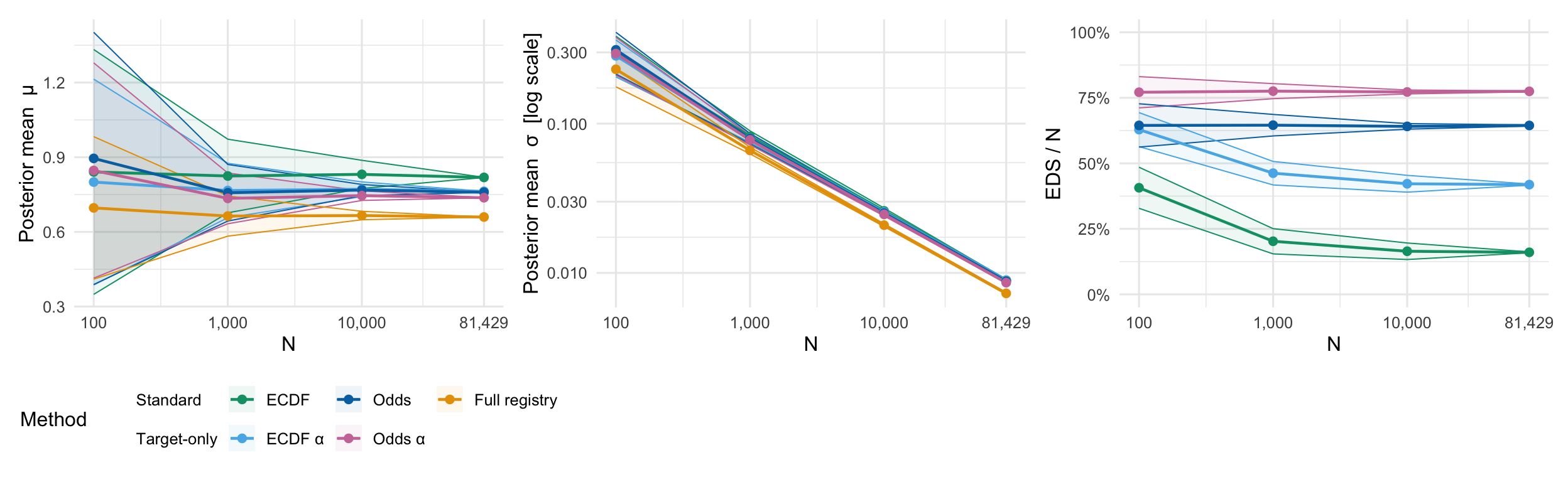}
    \caption{Results for the analysis trained on the full observed registry data $g(X)$, to predict the effect size in the general population $q(X)$. Statistics plotted are posterior means ($\mu$), SDs ($\sigma$, and EDS values, averaged over random seeds with confidence intervals. Different inference methods are represented with different data subsample sizes N.}
\label{fig7}    
\end{figure}

\begin{table}[!h]
\centering
\caption{\label{tab:rd_fp}Full-population projection: posterior means of $\mu$ and $\sigma$ across 20 repeated subsamples (mean and SD across seeds). \textit{Full registry}: standard Bayes on the complete registry; serves as the reference since no ground truth exists for this task. EDS shown as mean (SD) across seeds. $<.001$ indicates uncertainty $< .001$.}
\centering
\resizebox{\ifdim\width>\linewidth\linewidth\else\width\fi}{!}{
\begin{tabular}[t]{llll}
\toprule
Method & $\mu$ (SD) & $\sigma$ (SD) & EDS (SD)\\
\midrule
\addlinespace[0.4em]
\multicolumn{4}{l}{$N = 100$}\\
\hspace{1em}Full registry & .697 (.286) & .2312 (.0547) & ---\\
\hspace{1em}ECDF & .841 (.492) & .2985 (.0866) & 41 (8)\\
\hspace{1em}ECDF $\alpha$ & .800 (.413) & .2845 (.0790) & 63 (6)\\
\hspace{1em}Odds & .895 (.507) & .3120 (.0972) & 64 (8)\\
\hspace{1em}Odds $\alpha$ & .847 (.432) & .2947 (.0838) & 77 (6)\\
\addlinespace[0.4em]
\multicolumn{4}{l}{$N = 1000$}\\
\hspace{1em}Full registry & .664 (.081) & .0662 (.0034) & ---\\
\hspace{1em}ECDF & .825 (.148) & .0820 (.0075) & 203 (48)\\
\hspace{1em}ECDF $\alpha$ & .766 (.109) & .0799 (.0057) & 462 (45)\\
\hspace{1em}Odds & .757 (.114) & .0795 (.0064) & 646 (41)\\
\hspace{1em}Odds $\alpha$ & .735 (.102) & .0778 (.0057) & 775 (29)\\
\addlinespace[0.4em]
\multicolumn{4}{l}{$N = 10000$}\\
\hspace{1em}Full registry & .666 (.017) & .0209 (.0005) & ---\\
\hspace{1em}ECDF & .831 (.057) & .0260 (.0012) & 1,646 (316)\\
\hspace{1em}ECDF $\alpha$ & .771 (.029) & .0256 (.0008) & 4,221 (321)\\
\hspace{1em}Odds & .768 (.022) & .0252 (.0006) & 6,409 (107)\\
\hspace{1em}Odds $\alpha$ & .746 (.019) & .0247 (.0006) & 7,722 (76)\\
\addlinespace[0.4em]
\multicolumn{4}{l}{$N = 81429$}\\
\hspace{1em}Full registry & .660 ($<$.001) & .0073 (.0001) & ---\\
\hspace{1em}ECDF & .819 ($<$.001) & .0089 (.0001) & 13,063 (142)\\
\hspace{1em}ECDF $\alpha$ & .763 ($<$.001) & .0089 (.0002) & 34,101 (186)\\
\hspace{1em}Odds & .759 (.001) & .0088 (.0001) & 52,481 (201)\\
\hspace{1em}Odds $\alpha$ & .737 ($<$.001) & .0086 (.0001) & 63,081 (143)\\
\bottomrule
\end{tabular}}
\end{table}

\begin{table}[ht]
\centering
\caption{Empirical 90\% credible interval coverage before (Raw) and after sandwich correction (SW), together with the mean correction factor $K$, for each IPW method averaged over 100 seeds.  Standard Bayes and Target-only are omitted; see text.  Nominal target is 0.90.}
\label{tab:sandwich}
\begin{tabular}{@{}llrrr@{}}
\toprule
Parameter & Method & Raw & SW & Mean $K$ \\
\midrule
\multicolumn{5}{@{}l}{\textit{1D Gaussian}\; ($n_{\mathrm{reg}}=200$, $n_{\mathrm{clin}}=50$)} \\
\addlinespace[0.3em]
\multirow{3}{*}{$\mu$}
  & Oracle        & .89 & .91 & 1.01 \\
  & Odds          & .62 & .70 & 1.41 \\
  & Odds $\alpha$ & .61 & .69 & 1.32 \\
\addlinespace[0.3em]
\multirow{3}{*}{$\sigma$}
  & Oracle        & .95 & .92 & 0.87 \\
  & Odds          & .53 & .48 & 1.30 \\
  & Odds $\alpha$ & .36 & .32 & 1.22 \\
\midrule
\multicolumn{5}{@{}l}{\textit{4D Regression}\; ($n_{\mathrm{reg}}=500$, $n_{\mathrm{clin}}=200$)} \\
\addlinespace[0.3em]
\multirow{3}{*}{Intercept}
  & Oracle        & .56 & .71 & 1.38 \\
  & Odds          & .67 & .72 & 1.22 \\
  & Odds $\alpha$ & .58 & .60 & 1.12 \\
\addlinespace[0.3em]
\multirow{3}{*}{Biomarker}
  & Oracle        & .61 & .79 & 1.64 \\
  & Odds          & .49 & .63 & 1.54 \\
  & Odds $\alpha$ & .48 & .60 & 1.34 \\
\addlinespace[0.3em]
\multirow{3}{*}{Age}
  & Oracle        & .54 & .77 & 1.62 \\
  & Odds          & .48 & .62 & 1.52 \\
  & Odds $\alpha$ & .44 & .54 & 1.32 \\
\addlinespace[0.3em]
\multirow{3}{*}{Sex}
  & Oracle        & .63 & .87 & 1.65 \\
  & Odds          & .66 & .78 & 1.36 \\
  & Odds $\alpha$ & .71 & .80 & 1.19 \\
\addlinespace[0.3em]
\multirow{3}{*}{$\sigma$}
  & Oracle        & .49 & .77 & 1.60 \\
  & Odds          & .36 & .46 & 1.40 \\
  & Odds $\alpha$ & .34 & .41 & 1.19 \\
\bottomrule
\end{tabular}
\end{table}

\subsection{Out-of-Sample Validation}

For the first situation, we intend to demonstrate the ability of the reweighting methods to recover accurate posterior when injecting artificial bias into real data. We selected a random subsample of the observed data set and defined it as the unbiased population for the purposes of having an observed true distribution $G$. We subsequently constructed an artificial subsample from a biased distribution $H$, by sampling individuals with replacement from the unbiased population, according to a probability proportional to the empirical CDF of the age of the unbiased data $G$, such that older populations were more likely to be selected. We then compared models run on the unbiased population to those run on the biased population (``Standard'') with different reweighting techniques. The reweighting techniques include the GLM-derived weights (``Odds'') , those derived from the empirical cumulative distribution functions of the distributions (``ECDF'' accessible here since we are comparing using a single, continuous variable), and including an $\alpha=0.75$ tempering of the weights(``Odds $\alpha$'' and ``ECDF $\alpha$'') .

The procedure was iterated over random seeds, and we further derive biases and RMSEs relative to the statistics derived from the full, unbiased population (``Full-registry''). We consider replicating the mean posterior values ($\mu$) of the unbiased sample to be desirable, but since the models are trained with differing amounts of data, we do not necessarily expect the posteriors to exhibit the same sd ($\sigma$).  We use it here as a point of reference to test the accuracy and consistency of the classifier-trained method. We consider differing data scales, where the unbiased data set consists of 100, 1000, 10000 and the full 81429 data points from the registry data.

Results are presented in Figure \ref{fig6}. We see that for the posterior mean, the RMSE and bias decrease for all methods (except Standard) for increasing amounts of data. All of the methods (except for ECDF with larger N) exhibited a negative bias, implying a tendency to undercorrect. The $\alpha$-tempering of the reweighting appeared to benefit for smaller N, but not for the larger data sets. The GLM-derived weights provided the lowest-RMSE posteriors in all but the smallest data sets, suggesting that once there is enough data to train the classifiers adequately then they provide desirable results.

The results concerning the posterior sd $\sigma$ are harder to interpret, as it is not clear that the posteriors trained on a reweighted sub-sampled would be expected to converge in distribution to those trained on the full data set. However, some trends are clear: increasing the amount of data reduced the mean and variation in $\sigma$ for all methods, and all methods exhibited a positive bias relative to the full data analysis. The $\alpha$ tempering consistently reduced the mean $\sigma$ but increased the variation, while the GLM-derived weighting procedures often exhibited smaller $\sigma$s with less variation compared to the ECDF-derived weights.

We see that the EDS increased with the amount of data. However, for increasing $N$ the EDS decreased relative to total number of data points available, without ever collapsing to a very small number. Including the tempering $\alpha$ increased the EDS, and the GLM-derived weights consistently gave a higher and more consistent EDS compared to the ECDF-derived weights.

\subsection{Full Population Projection}

The second situation involves treating the observed data set from the distribution $G$ as a biased sample  from the unobserved true distribution $P$ and projecting down to the demographic distribution of the general-population $Q$ of all men 39 years of age or older, as previously described. There is no observed ground truth in this context, so we will restrict ourselves to comparing the properties of the posteriors generated by the reweighting methods. 

We report the results visually in Figure \ref{fig7} and numerically in Table \ref{tab:rd_fp}. In the absence of a ground truth for this analysis, we present means and standard deviation of the posterior mean and sd, and the EDS. We see that the variation on the posterior means and SDs decreases with increasing amounts of data, and the EDS increases. Unlike the previous example, we do not see the EDS decrease relative to the total amount of data.

Relative to the baseline ``Full registry'' estimate trained on the biased registry data, we see that the GLM-derived ``Odds'' weights tend to provide more conservative updates than the ECDF weights, and that the  $\alpha$-tempering further decreases the size of the update. The GLM-derived weights thus provide the tightest bounds on the unknown true odds ratio, and highest EDS, with no indication to suggest that they should be less accurate than the other methods.

 \section{Discussion}\label{sec:discussion}

We consider the results presented to be a convincing empirical representation of the applicability of the IPW-style reweighting methods developed for Bayesian inference for large-scale practical problems. Some uncertainty remains concerning optimal design choices for methodological components, and for general applicability to practical problems outside of those considered here.

There is space for refinement concerning classifier choice and design for weight estimation: out-of-the-box GLM models work adequately for the lower-dimensional problems considered here, but the applicability to higher-dimensional parameter spaces could be explored through the use of regularised logistic regression models, or more complex machine learning classifiers. The question of weight stability (present in the EDS statistic) appears most concerning in smaller-data domains: we might expect a simpler parametric classifier to make best use of small amounts of data, conditional on its parametrisation being well-specified. The consequences of classifier misspecification and the potential resulting inferential bias would be of interest to explore in future work. Similarly, while we have provided theoretical results concerning noise introduced by using classifiers to estimate density ratios, the empirical consequences of exactly how much noise is admissible for which practical problems are still not fully understood. It is also interesting and possibly surprising that the classifier-derived weights appear to outperform those derived from the observed ECDFs of the data distributions: it is possible that the classifiers impose regularity or smoothness on the density ratio modelling, which helps stabilise the derived weights.

Further, details of how to practically process the weights are also of interest: the clipping of the extreme values (we have restricted log-odds of $\pm10$ here), or choice of the weight tempering constant ($\alpha=0.75$ here) may vary from problem to problem. We also feel that use of cross-validation to estimate the weights is a good idea in the data domains we have considered, but this may become prohibitive computationally in larger data sets, and may have to be abandoned at the expense of potentially introducing bias. The question also remains as to whether the same data can be used for training the weight-estimating classifiers and for performing the belief update: we think this is a defensible choice in our examples, but further robustness may be obtained by separating these inference steps entirely.

We have kept the inference machinery within the well-established brms software: this is a very powerful HMC-driven framework encompassing many known parametric Bayesian regression models, but the question remains as to how these will generalise to other model specifications or inference methods. Combining with Sequential Monte Carlo would be of interest, considering the shared reweighting logic and challenges of potential collapse into a small number of samples. The relationship with variational methods is likely to be fruitful, as both methods rely on KL-divergences to motivate inferential choices. Further extensions to Stochastic Variational Inference and streaming methods would raise further questions about how to learn the weights when data is received on-the-fly. The extension to non-i.i.d. models such as network models or time series would also be of interest, since the theory concerning the independent projection of the contributions to the KL divergences may not translate automatically.

We consider this to be a demonstration of generalised Bayesian reasoning extended to problems outside of the traditional concerns of model misspecification. Skepticism of statistical practitioners concerning post-Bayesian models can be dispelled considering the well-known biases of standard Bayesian updates in the presence of selection bias, and the theoretical results here demonstrating unbiased convergence (conditional on successful training of the classifiers), while maintaining other attractive features of Bayesian updating. Given that the machinery of IPW is very mature in the frequentist literature, we consider there to be clear extensions to other potential sources of biases, including misspecification and causality, but we reserve these extensions for future work.

\section{Conclusion}\label{sec:conclusion}

We consider the results presented to be a successful application of the methods developed for realistic problems with simulated and real data. The theory and practice of Bayesian inference and frequentist IPW are mature and active research fields: this work provides the theoretical and empirical justification to combine the two. The methods are also immediately applicable using widely used software: here we use off-the-shelf GLM classifiers and the brms inference package, but they can certainly be applied to more complex, application-tailored specifications of either the classifier or the statistical model. This work opens the door to the use of IPW within contemporary post-Bayesian inference approaches.

\section{Acknowledgements}\label{sec:acknowledgements}

We would like to thank Jan Oldenburg and Johan Bjerner of Akershus University Hospital and Fürst medical Laboratory for collecting the prostate cancer data set, granting permission for its use in the article, and giving feedback on the manuscript. We also thank Soili Marianne Lehto for motivating our initial interest in selection bias as a practical challenge.We express further gratitude to Leiv Rønneberg for comments on the idea and manuscript.

\begin{funding}
The authors have no funding information to declare.
\end{funding}

\newpage

\begin{supplement}
\section{Classifier Pragmatics}\label{app:classprag}
In practice the empirically trained classifier will inject some noise into the inference process and thus potentially introduce instabilities. We have a few recommendations to help address these challenges:

\begin{itemize}
\item We recommend clipping the maximum and minimum values of the weights to avoid collapse of the update to rely on a very small number of data points (e.g. log odds minimum and maximum values of -10 and 10, respectively).
\item It may be desirable to normalise the weights before use, such that $\sum_{i=1}^n w_i = n$, to help ensure numerical stability during the update itself. 
\item If \(r\) is unbounded or poorly estimated, tempered odds \(r^\alpha\) (\(0<\alpha<1\)) can stabilise inference at the cost of changing the target distribution (you then project onto the density proportional to \(q(x)^{\alpha}g(x)^{1-\alpha}\) rather than \(q\)).
\item We recommend always reporting EDS and weight histograms, to verify a sufficiently diverse update not dominated by a very small number of observed data points.
\item Cross-fitting can reduce bias from evaluating fitted weights on the observations used to train them. Its effect on asymptotic uncertainty requires analysis of the particular estimator and training scheme.
\end{itemize}

\section{Setup and Notation for Theorem Proofs}\label{App:setup}

Let \(X_1,\dots,X_n\) be i.i.d.\ draws from \(G\) with density \(g\) (w.r.t.\ a dominating measure \(\mu\)).
Throughout, $\R$ denotes the real numbers, and $\Pp$ and $\E$
denote probability and expectation under the joint sampling law
of the observations and, where applicable, any auxiliary sample
used to estimate the weights. Subscripts indicate integration
with respect to a specified distribution; for example,
\[
    \E_G[f(X)] = \int f(x)g(x)\,d\mu(x),
\]
whenever the expectation exists, with analogous notation for
$\E_Q$. A vertical bar in a probability, density, or expectation
denotes conditioning.
Let \(\mathcal P=\{p_\theta:\theta\in\Theta\}\) be a parametric family with log-density \(\ell_\theta(x)=\log p_\theta(x)\). Fix a measurable, nonnegative weight function \(w:\mathcal X\to[0,\infty)\). Define the empirical weighted risk
\[
\mathcal L_n(\theta;w) \;=\; \frac{1}{n}\sum_{i=1}^n w(X_i)\,(-\ell_\theta(X_i)),
\]
and population risk
\[
\mathcal L(\theta;w) \;=\; \mathbb{E}_G[w(X)(-\ell_\theta(X))] = \int -\ell_\theta(x)\,w(x)g(x)\,d\mu(x).
\]
When \(w(x)\propto q(x)/g(x)\) this equals \(\mathbb{E}_Q[-\ell_\theta(X)]\).

Let \(\theta^\star\in\arg\min_{\theta\in\Theta}\mathcal L(\theta;w)\) denote a population minimiser (pseudo-true parameter). Scores and Hessians:
\[
s_\theta(x):=\nabla_\theta \ell_\theta(x),\qquad H_\theta(x):=-\nabla^2_\theta\ell_\theta(x).
\]

We will use standard empirical process notation see \citet{vdV98,vdVW96} for background. Write $\Pi(d\theta)=\pi(\theta)\,d\theta$ for the prior measure and
\[
    \Pi_n(d\theta)=
    \frac{\exp\{-n\mathcal L_n(\theta;w)\}\Pi(d\theta)}
         {\int_\Theta\exp\{-n\mathcal L_n(u;w)\}\Pi(du)}.
\]
Whenever a posterior is used, its normalising constant is assumed to be
finite and strictly positive almost surely. Thus,
$\mathcal L_n(\theta;w)=P_n[-w\ell_\theta]$ and
$\mathcal L(\theta;w)=P[-w\ell_\theta]$.
The associated empirical process is
\[
    \mathbb G_n(f)
    := \sqrt{n}(P_n-P)f
    = \frac{1}{\sqrt{n}}\sum_{i=1}^n
      \{f(X_i)-\E_G[f(X)]\}.
\]

We write $\xrightarrow{\Pp}$, $\xrightarrow{d}$, and
$\xrightarrow{a.s.}$ for convergence in probability, in
distribution, and almost surely, respectively, under the
sampling law described above. For a deterministic sequence
$a_n>0$ and random quantities $Z_n$, we write
$Z_n=\op(a_n)$ if
\[
    \Pp\bigl(\|Z_n\|>\varepsilon a_n\bigr)\longrightarrow 0
    \qquad\text{for every }\varepsilon>0,
\]
and $Z_n=\Op(a_n)$ if $Z_n/a_n$ is bounded in probability,
that is,
\[
    \lim_{M\to\infty}\limsup_{n\to\infty}
    \Pp\bigl(\|Z_n\|>Ma_n\bigr)=0.
\]
Here $\|\cdot\|$ denotes absolute value for scalars, the Euclidean
norm for vectors, and the induced operator norm for matrices.

\section{Assumptions}\label{app:assump}

Below the assumptions are stated in empirically checkable form (Donsker / GC / Lindeberg / moment bounds). In practice, for parametric models these reduce to routine conditions.

\begin{assump}[Parameter regularity]\label{A:param}
The parameter space $\Theta\subset\mathbb R^d$ contains an open
neighborhood $N$ of $\theta^\star$. For \(G\)-a.e.\ \(x\), \(\theta\mapsto \ell_\theta(x)\) is twice continuously differentiable in a neighbourhood \(N\) of \(\theta^\star\). The maps \(x\mapsto s_\theta(x)\) and \(x\mapsto H_\theta(x)\) are measurable for \(\theta\in N\).
\end{assump}

\begin{assump}[Identifiability and nonsingularity]\label{A:ident}
The function \(\theta\mapsto \mathcal L(\theta;w)\) has a unique minimiser \(\theta^\star\in N\). Define
\[
I := \mathbb{E}_G[w(X)H_{\theta^\star}(X)].
\]
Assume \(I\) is positive definite.
\end{assump}

\begin{assump}[Envelope and moments]\label{A:env}
There exists an integrable envelope \(M:\mathcal X\to[0,\infty)\) with \(\mathbb{E}_G[w(X)M(X)]<\infty\) such that for \(\theta\in N\),
\[
\|s_\theta(x)\| \le M(x), \qquad \|H_\theta(x)\|\le M(x),\quad \text{for }G\text{-a.e. }x.
\]
Also assume \(\mathbb{E}_G[w(X)^2\|s_{\theta^\star}(X)\|^2]<\infty\) and $\mathbb E_G[w(X)|\ell_{\theta^\star}(X)|]<\infty$.
\end{assump}

\begin{assump}[Uniform LLN / Glivenko--Cantelli]\label{A:GC}
The class \(\mathcal F=\{ x\mapsto w(x)\ell_\theta(x) : \theta\in N\}\) is \(G\)-Glivenko--Cantelli:
\[
\sup_{\theta\in N}\left| (\mathbb{P}_n-\mathbb{P})[w\ell_\theta]\right| \xrightarrow{a.s.} 0.
\]
A sufficient condition is: \(\mathcal F\) has an integrable envelope and finite bracketing entropy (see \citet[Ch.\ 2]{vdVW96}).
\end{assump}

\begin{assump}[CLT for weighted scores]\label{A:CLT}
The class \(\mathcal S=\{x\mapsto w(x)s_{\theta^\star}(x)^\top v : v\in\mathbb{R}^d, \|v\|=1\}\) is Donsker and
\[
\mathbb{G}_n\big(w s_{\theta^\star}\big) \xrightarrow{d} \mathcal{N}(0,J),
\]
with
\[
J = \mathbb{E}_G\big[w(X)^2 s_{\theta^\star}(X) s_{\theta^\star}(X)^\top\big] < \infty.
\]
Sufficient: \(\mathcal S\) is \(G\)-Donsker (e.g. finite VC / finite entropy) and second moments finite.
\end{assump}

\begin{assump}[Prior]\label{A:prior}
The prior $\Pi$ is a probability measure with a Lebesgue density $\pi$
that is continuous and strictly positive in a neighbourhood of $\theta^\star$.
\end{assump}

\begin{remark}
For parametric families with smooth log-densities and bounded parameter sets, Assumptions \ref{A:param}--\ref{A:CLT} are straightforward to verify. For nonparametric classifiers / weights, verifying Donsker conditions is more delicate; cross-fitting will be recommended to relax such requirements (see Section \ref{app:plugin} below and \citet{chernozhukov2018double}).
\end{remark}

\begin{assump}[Parametric ratio estimator]\label{A:ratio}
There exists a finite-dimensional parameter \(\gamma\in\Gamma\subset\mathbb{R}^p\) and a map \(r_\gamma(x)\) such that the true ratio \(r(x)=q(x)/g(x)=r_{\gamma_0}(x)\) for some \(\gamma_0\). The estimator \(\hat\gamma_m\) (possibly based on an independent sample of size \(m\)) satisfies
\[
\sqrt{m}(\hat\gamma_m-\gamma_0)\ \xrightarrow{d}\ N(0,\Sigma_\gamma).
\]
Further \(r_\gamma(x)\) is twice continuously differentiable in \(\gamma\) in a neighbourhood of \(\gamma_0\), and \(\mathbb{E}_G[\|\dot r_{\gamma_0}(X)\|^2\|s_{\theta^\star}(X)\|^2]<\infty\).
In addition, for some convex neighbourhood $U$ of $\gamma_0$, define
\[
    D_1(x)=\sup_{\gamma\in U}\|\dot r_\gamma(x)\|,
    \qquad
    D_2(x)=\sup_{\gamma\in U}\|\ddot r_\gamma(x)\|,
    \qquad
    M_H(x)=\sup_{\theta\in N}\|H_\theta(x)\|.
\]
Assume these envelopes are measurable and
\[
    \mathbb E_G[\|s_{\theta^\star}(X)\|D_2(X)]<\infty,
    \qquad
    \mathbb E_G[D_1(X)M_H(X)]<\infty.
\]

\end{assump}

\section{Existence and uniqueness of the pseudo-true parameter: Proof of Theorem \ref{thm:exist}}\label{app:exist}

\begin{proof}
We argue via dominated convergence for continuity, the extreme-value theorem for existence, and convexity for uniqueness.

Fix $\theta_0\in\Theta$ and let $\theta_j\to\theta_0$ in $\Theta$.
Eventually $\theta_j\in U_{\theta_0}$. Pointwise continuity and the
integrable envelope therefore give
$\mathcal L(\theta_j;w)\to\mathcal L(\theta_0;w)$ by dominated convergence.
Finiteness follows from the same envelope. The compact-set conclusion is
the extreme-value theorem. If $\Theta$ is closed and the risk is coercive,
the nonempty sublevel set
\[
    C_0=\{\theta\in\Theta:
        \mathcal L(\theta;w)\le\mathcal L(\theta_0;w)\}
\]
is closed and bounded, hence compact. A minimiser on $C_0$ is a global
minimiser on $\Theta$. Strict convexity on a convex parameter space
excludes two distinct minimisers by considering their midpoint.
\end{proof}
 
\section{Consistency (posterior concentration): Proof of Theorem \ref{thm:consistency}}\label{app:consistency}

\begin{proof}
We construct a high-probability uniform-convergence event under which the posterior mass on the complement of any $\varepsilon$-ball around $\theta^\star$ is bounded by an exponentially decaying ratio.

Fix $\epsilon>0$. If $A_\epsilon$ is empty, the conclusion is immediate.
Otherwise let $\Delta=\Delta_\epsilon>0$.
 
By continuity at $\theta^\star$, choose $0<\delta<\epsilon$ such that
\[
    \sup_{\theta\in B_\delta}
       \{\mathcal L(\theta;w)-\mathcal L(\theta^\star;w)\}
    <\frac{\Delta}{8},
    \qquad
    B_\delta:=\{\theta\in\Theta:
                    \|\theta-\theta^\star\|<\delta\}.
\]
Since $\pi$ is continuous and strictly positive near $\theta^\star$,
$\Pi(B_\delta)>0$ after decreasing $\delta$ if necessary. Define the event
\[
    E_n=
    \left\{
       \sup_{\theta\in\Theta}
       |\mathcal L_n(\theta;w)-\mathcal L(\theta;w)|<\frac{\Delta}{8}
    \right\}.
\]
By the definition of the weighted risk, the $G$-Glivenko--Cantelli property of
$\mathcal F_\Theta$ gives
\[
    \sup_{\theta\in\Theta}
       |\mathcal L_n(\theta;w)-\mathcal L(\theta;w)|
    =
    \sup_{\theta\in\Theta}
       |(P_n-P)[w\ell_\theta]|
    \xrightarrow{\Pp}0.
\]
Hence $\Pp(E_n)\to1$.

On $E_n$, every $\theta\in A_\epsilon$ satisfies
\[
    \mathcal L_n(\theta;w)
    \geq \mathcal L(\theta;w)-\frac{\Delta}{8}
    \geq \mathcal L(\theta^\star;w)+\frac{7\Delta}{8},
\]
whereas every $\theta\in B_\delta$ satisfies
\[
    \mathcal L_n(\theta;w)
    \leq \mathcal L(\theta;w)+\frac{\Delta}{8}
    \leq \mathcal L(\theta^\star;w)+\frac{\Delta}{4}.
\]
Consequently, on $E_n$,
\begin{align*}
    \int_{A_\epsilon}e^{-n\mathcal L_n(\theta;w)}\Pi(d\theta)
    &\leq
      \exp\left[-n\left\{\mathcal L(\theta^\star;w)
                         +\frac{7\Delta}{8}\right\}\right],\\
    \int_{B_\delta}e^{-n\mathcal L_n(\theta;w)}\Pi(d\theta)
    &\geq
      \Pi(B_\delta)
      \exp\left[-n\left\{\mathcal L(\theta^\star;w)
                         +\frac{\Delta}{4}\right\}\right].
\end{align*}
Using the integral over $B_\delta$ as a lower bound for the posterior
normalizing constant gives
\[
    \Pi_n(A_\epsilon)
    \leq
    \frac{1}{\Pi(B_\delta)}
    \exp\left(-\frac{5n\Delta}{8}\right)
    \longrightarrow0
    \qquad\text{on }E_n.
\]
Since $\Pp(E_n)\to1$, the displayed posterior probability converges to zero in
probability.
\end{proof}

\section{Asymptotic normality: sandwich variance and BvM: Proof of Theorem \ref{thm:asymp}}\label{app:asymp}


We give an M-estimation-based proof of the asymptotic distribution of the minimiser and establish the corresponding local Gaussian approximation of the generalised posterior. Because the weighted likelihood defines a generalised rather than, in general, an ordinary likelihood-based posterior, the Bernstein--von Mises argument requires regularity conditions for generalised posteriors rather than following directly from the misspecified-likelihood result of \citet{kleijn2012bernstein}. We refer to \citet{miller2021asymptotic} as further background on asymptotic results for generalised Bayesian posteriors.

\begin{proof}
The proof proceeds in two parts: (i) a Taylor-expansion argument establishes frequentist sandwich normality of the M-estimator; (ii) a second expansion around the M-estimator establishes local quadratic shape of the log-posterior, yielding the BvM approximation.

\textbf{(i) Asymptotic normality of $\hat\theta_n$.}
The local integrability and differentiability assumptions justify
differentiation under the expectation. Since $\theta^\star$ is an interior
minimiser,
\[
    \mathbb E_G[w(X)s_{\theta^\star}(X)]=0,
    \qquad
    \sqrt n\,\nabla_\theta\mathcal L_n(\theta^\star;w)
    =-\mathbb G_n(ws_{\theta^\star})
    \xrightarrow{d}\mathcal N(0,J).
\]
On the event that the segment from $\theta^\star$ to $\hat\theta_n$ lies
in a convex neighbourhood contained in $N$, define
\[
    B_n=\int_0^1\nabla_\theta^2\mathcal L_n
        \bigl(\theta^\star+t(\hat\theta_n-\theta^\star);w\bigr)\,dt.
\]
This event has probability tending to one. The fundamental theorem of
calculus gives the exact identity
\[
    0=\nabla_\theta\mathcal L_n(\theta^\star;w)
      +B_n(\hat\theta_n-\theta^\star).
\]
The Hessian envelope and pointwise continuity imply continuity of
$I(\theta)=\mathbb E_G[w(X)H_\theta(X)]$ near $\theta^\star$.
Uniform Hessian convergence and consistency of $\hat\theta_n$ therefore
give $B_n\xrightarrow{\Pp}I$. Since $I$ is positive definite, Slutsky's
lemma yields
\[
    \sqrt n(\hat\theta_n-\theta^\star)
    =-B_n^{-1}\sqrt n\,\nabla_\theta\mathcal L_n(\theta^\star;w)
    \xrightarrow{d}\mathcal N(0,I^{-1}JI^{-1}).
\]

\medskip\noindent
\textbf{(ii) Local posterior Gaussian approximation.}
For each fixed $M<\infty$, Taylor expansion in integral form gives,
uniformly over $\|h\|\le M$,
\[
\begin{split}
    n\{\mathcal L_n(\hat\theta_n+h/\sqrt n;w)
        -\mathcal L_n(\hat\theta_n;w)\}
    &=\int_0^1(1-t)h^\top\nabla_\theta^2\mathcal L_n
        (\hat\theta_n+th/\sqrt n;w)h\,dt\\
    &=\tfrac12h^\top I h+\op(1).
\end{split}
\]
Also $\pi(\hat\theta_n+h/\sqrt n)\to\pi(\theta^\star)>0$ uniformly
on each such ball, in probability.

We verify that the normalising constant and tails are controlled.
Choose $\delta>0$ sufficiently small. Continuity and positive definiteness
of $I(\theta)$, uniform Hessian convergence, and consistency imply that,
with probability tending to one, for fixed $0<\lambda<\Lambda<\infty$,
\[
    \frac{\lambda}{2}\|\theta-\hat\theta_n\|^2
    \le\mathcal L_n(\theta;w)-\mathcal L_n(\hat\theta_n;w)
    \le\frac{\Lambda}{2}\|\theta-\hat\theta_n\|^2
    \quad\text{if }\|\theta-\hat\theta_n\|\le\delta.
\]
The prior density is bounded above and below by positive constants on
this neighbourhood. Consequently, for some fixed $b>0$,
\[
    D_n:=\int_\Theta
       e^{-n\{\mathcal L_n(\theta;w)-\mathcal L_n(\hat\theta_n;w)\}}
       \Pi(d\theta)\ge b n^{-d/2}
\]
with probability tending to one, by integrating over
$\|\theta-\hat\theta_n\|\le n^{-1/2}$.

The global GC condition and consistency imply
$\mathcal L_n(\hat\theta_n;w)\xrightarrow{\Pp}
\mathcal L(\theta^\star;w)$. Separation at radius $\delta/2$ then gives
a fixed $\eta>0$ such that, with probability tending to one,
\[
    \inf_{\|\theta-\hat\theta_n\|>\delta}
       \{\mathcal L_n(\theta;w)-\mathcal L_n(\hat\theta_n;w)\}
    \ge\eta,
\]
whenever the set is nonempty. Splitting the integral at radius $\delta$
and using the preceding bounds gives, for a fixed finite $C$,
\[
    \Pi_n\{\sqrt n\|\theta-\hat\theta_n\|>M\}
    \le C\int_{\|h\|>M}e^{-\lambda\|h\|^2/2}\,dh
       +b^{-1}n^{d/2}e^{-n\eta}
\]
on events whose probabilities tend to one.
Thus the transformed posterior tails vanish as first $n\to\infty$ and
then $M\to\infty$. Combining this bound with the uniform expansion on
each fixed ball shows that the transformed posterior density converges
in $L^1$, in probability, to the $\mathcal N(0,I^{-1})$ density.
This implies the stated approximation over Borel subsets of every fixed
ball.
\end{proof}

\section{Plug-in (estimated) weights effect: Proof of Theorem \ref{thm:plugin}}\label{app:plugin}

We consider the case where \(w\) is unknown and estimated by \(\hat w_n\), typically by training a classifier to distinguish samples from \(Q\) and \(G\), or by direct density-ratio estimation. We give a rigorous result assuming a parametric estimator for the ratio; we then comment on nonparametric and cross-fitting remedies.

\begin{proof}
A first-order Taylor expansion of $r_{\hat\gamma_m}$ around $\gamma_0$ introduces an additional term in the score equation; the LLN for the Jacobian matrix and independence of the auxiliary sample then allow a joint CLT to be applied.

Write $\delta_m=\hat\gamma_m-\gamma_0$ and
\[
    S_n(\theta;\hat r_n)
    =-P_n[\hat r_n s_\theta],\qquad
    A_n=P_n[s_{\theta^\star}\dot r_{\gamma_0}^\top].
\]
The oracle first-order condition is $P[rs_{\theta^\star}]=0$.
Taylor expansion of $r_\gamma$, with the added envelope in
Assumption~\ref{A:ratio}, therefore gives
\[
    \sqrt n S_n(\theta^\star;\hat r_n)
    =-\mathbb G_n(rs_{\theta^\star})-A_n\sqrt n\delta_m+R_n,
\]
where, on the event $\hat\gamma_m\in U$,
\[
    \|R_n\|
    \le\frac{\sqrt n}{2}\|\delta_m\|^2
       P_n[\|s_{\theta^\star}\|D_2]
    =\Op(\sqrt n/m)=\op(1).
\]
The last equality follows because $m\to\infty$ and $n/m$ remains bounded.
The law of large numbers gives $A_n\xrightarrow{\Pp}A$, and
\[
    \sqrt n\delta_m
    =\sqrt{n/m}\sqrt m\delta_m
    \xrightarrow{d}\sqrt c Z_\gamma,
    \qquad Z_\gamma\sim\mathcal N(0,\Sigma_\gamma).
\]
The oracle score CLT and independence of the auxiliary sample imply
\[
    \sqrt n S_n(\theta^\star;\hat r_n)
    \xrightarrow{d}\mathcal N(0,J+cA\Sigma_\gamma A^\top).
\]

For the empirical Hessian, the second added envelope gives
\[
    \sup_{\theta\in N}
    \|P_n[(\hat r_n-r)H_\theta]\|
    \le\|\delta_m\|P_n[D_1M_H]=\op(1)
\]
with probability tending to one. Together with the oracle Hessian GC
condition, this implies uniform convergence to
$I_r(\theta)=\mathbb E_G[r(X)H_\theta(X)]$ on $N$.
Define the integrated Hessian along the segment to the consistent root by
\[
    \widetilde B_n=
    \int_0^1\nabla_\theta^2\mathcal L_n
       (\theta^\star+t(\tilde\theta_n-\theta^\star);\hat r_n)\,dt.
\]
This is well defined on an event of probability tending to one, and
$\widetilde B_n\xrightarrow{\Pp}I_r(\theta^\star)=I$.
The fundamental theorem of calculus gives
\[
    0=S_n(\theta^\star;\hat r_n)
       +\widetilde B_n(\tilde\theta_n-\theta^\star).
\]
Slutsky's lemma therefore yields
\[
    \sqrt n(\tilde\theta_n-\theta^\star)
    \xrightarrow{d}
    \mathcal N\bigl(0,I^{-1}(J+K_c)I^{-1}\bigr).
\]
For $c=0$ this is the oracle sandwich limit, while $0<c<\infty$ gives
the additional weight-estimation variance.
\end{proof}

\begin{remark}
Theorem~\ref{thm:plugin} concerns an independently trained parametric
ratio estimator. If the inference observations also enter ratio training,
the score and ratio estimator may have correlated first-order limits.
Cross-fitting reduces own-observation dependence but does not generally
remove the first-order contribution of ratio estimation, since $A$ need
not vanish. Applying an analogous result to a cross-fitted classifier
requires a joint asymptotic analysis for its actual training scheme.
\end{remark}

\section{Conditions on $G$ and $Q$ for feasible ratio estimation}\label{app:gq}

We summarise necessary / sufficient practical conditions:

\begin{enumerate}
\item \textbf{Overlap / positivity:} \(Q\ll G\) (i.e.\ \(r(x)=q(x)/g(x)\) exists finite \(G\)-a.e.). If \(q>0\) on sets where \(g=0\) no estimator from \(G\)-samples can recover \(Q\) there.

\item \textbf{Finite weighted moments:} Ensure \(\mathbb{E}_G[r^2\|s_{\theta^\star}\|^2]<\infty\) and \(\mathbb{E}_G[r\|H_{\theta^\star}\|]<\infty\). These guarantee the LLN/CLT used above. If these fail, importance weights have infinite variance and EDS collapses.

\item \textbf{Classifier learnability and calibration:} Using a classifier to estimate \(\eta(x)=\Pr(X\sim Q\mid x)\) gives ratios via \(\eta/(1-\eta)\) up to sampling-proportion constants. For reliable odds one needs well-calibrated probabilities: train with log-loss and (optionally) recalibrate (Platt/isotonic). For nonparametric classifiers use cross-fitting to relax Donsker-type requirements; see \citet{chernozhukov2018double}.

\item \textbf{Tail control / tempering:} If \(r\) is heavy-tailed so that \(\mathbb{E}_G[r^2]=\infty\), use clipping \(r\mapsto \min(r,\tau)\) or temper \(r^\alpha\) with \(\alpha\in(0,1)\) to ensure finite second moments and stabilize EDS. Note this alters the target distribution (the posterior then targets the KL projection of the modified target).
\end{enumerate}
\end{supplement}

\clearpage

\bibliographystyle{ba}
\bibliography{references}

\end{document}